\documentclass[a4paper,10pt,openany]{amsbook}
\usepackage{amsmath,amssymb,amsthm}
\usepackage{mathptmx}
\DeclareMathAlphabet{\mathcal}{OMS}{cmsy}{m}{n}

\allowdisplaybreaks[2]

\usepackage{times}
\usepackage{appendix}
\usepackage{stmaryrd}
\usepackage{graphicx}
\usepackage{psfrag}

\title{Properties and Construction of Polar Codes}
\author{Ryuhei Mori\\\vspace{1em} \large Supervisor: Toshiyuki Tanaka}
\address{Department of Systems Science, Graduate School of Informatics, Kyoto University}
\date{1 February, 2009}

\newtheorem{theorem}{Theorem}
\newtheorem{definition}[theorem]{Definition}
\newtheorem{lemma}[theorem]{Lemma}
\newtheorem{corollary}[theorem]{Corollary}
\newtheorem{proposition}[theorem]{Proposition}

\numberwithin{theorem}{chapter}
\numberwithin{section}{chapter}
\numberwithin{equation}{chapter}

\begin{document}
\frontmatter
\begin{titlepage}
\begin{center}
{\LARGE\bf
Properties and Construction of Polar Codes
}

\vspace{3em}
{\Large
Ryuhei Mori
}

\vspace{2em}
{\large
Supervisor: Toshiyuki Tanaka
}

\vspace{30em}
{
\scshape
Department of Systems Science, Graduate School of Informatics

Kyoto University
}

\vspace{2em}
{
\large
Feb 1, 2010
}
\end{center}
\end{titlepage}
\setcounter{page}{0}
\thispagestyle{empty}
\mbox{}\newpage
\chapter*{Abstract}
{\normalsize
Recently, Ar{\i}kan introduced the method of channel polarization on which one can construct
efficient capacity-achieving codes, called polar codes, for any binary discrete memoryless channel.
In the thesis, we show that decoding algorithm of polar codes, called successive cancellation decoding,
can be regarded as belief propagation decoding, which has been used for decoding of low-density parity-check codes, on a tree graph.
On the basis of the observation, we show an efficient construction method of polar codes using density evolution,
which has been used for evaluation of the error probability of belief propagation decoding on a tree graph.
We further show that channel polarization phenomenon and polar codes can be generalized
to non-binary discrete memoryless channels.
Asymptotic performances of non-binary polar codes, which use non-binary matrices called the Reed-Solomon matrices,
are better than asymptotic performances of the best explicitly known binary polar code.
We also find that
the Reed-Solomon matrices are considered to be natural generalization of the original binary channel polarization introduced by Ar{\i}kan.
}

\chapter*{Acknowledgment}
{\normalsize
I would like to thank the supervisor Toshiyuki Tanaka for insightful suggestions and creative comments.
I thank all members of the laboratory for their encouragement, and friendship.
}

\tableofcontents

\mainmatter
\chapter{Introduction}
\section{Overview}
The channel coding problem, in which one attempts to realize reliable communication on an unreliable channel,
is one of the most central problems of information theory.
Although it had been considered that unlimited amounts of redundancy is needed for reliable communication,
Shannon showed that in large systems, one only has to pay limited amounts of redundancy for reliable communication~\cite{shannon1948mtc}.
Shannon's result is referred to as ``the channel coding theorem''.
Although the theorem shows the existence of a good channel code,
we have to explicitly find desired codes and practical encoding and decoding algorithms,
in order to realize reliable efficient communication.

\section{Channel Model and Channel Coding Problem}
\subsection{Channel model}
Let $\mathcal{X}$ and $\mathcal{Y}$ denote sets of input and output alphabets.
Assume that $\mathcal{X}$ is finite and that $\mathcal{Y}$ is at most countable.
A discrete memoryless channel $W$ is defined as conditional probability distributions
$W(y \mid x)$ of $y\in\mathcal{Y}$ for all $x\in\mathcal{X}$
which represent probability that a channel output is $y$ when $x$ is transmitted.

\subsection{Channel coding problem}
Let $\mathcal{M}$ and $M$ denote a set of messages and its cardinality, respectively.
When $M = |\mathcal{X}|$, we can make a one-to-one correspondence between $\mathcal{M}$ and $\mathcal{X}$.
Let us consider communication where a sender transmits $x\in\mathcal{X}$ which represents a corresponding message $m\in\mathcal{M}$
and a receiver estimates $m$ (equivalently $x$) from received alphabet $y$.
Let $\psi(y)\in\mathcal{M}$ denote an estimation given $y$.
In this communication, an \textit{error probability of a channel} $W$ is 
\begin{equation*}
\frac1M \sum_{x\in\mathcal{X}}\sum_{y\in\mathcal{Y}} W(y\mid x)
\mathbb{I}\{\psi(y)\ne x\}
\end{equation*}
where $\mathbb{I}$ is the indicator function.

When an error probability of $W$ is larger than desired even if an estimator $\psi(y)$ is optimal,
we have to consider using a channel $W$ multiple times
in order to improve reliability of communication.
Let $z_0^{n-1}$ denote a vector $(z_0,\dotsc,z_{n-1})$ and $z_i^j$ denote subvector $(z_i,\dotsc,z_j)$ of $z_1^n$.
If one sends $x_0^{n-1}\in\mathcal{X}^n$ by using a channel $W$ $n$ times,
we assume that the transition probability is $W^n(y_0^{n-1} \mid x_0^{n-1}) := \prod_{i=0}^{n-1} W(y_i\mid x_i)$ 
for all $y_0^{n-1}\in\mathcal{Y}^n$.
This property of channel is referred to as \textit{memoryless}.

Mappings $\phi: \mathcal{M} \to \mathcal{X}^n$ and $\psi: \mathcal{Y}^n \to \mathcal{M}$ denote \textit{encoder} and \textit{decoder}, respectively
for some $n\in\mathbb{N}$ called the \textit{blocklength}.
An image of $\phi$ and its elements are called \textit{code} and \textit{codewords}, respectively.
An \textit{error probability of a code} is defined as 
\begin{equation*}
\frac1M\sum_{a\in\mathcal{M}}\sum_{y_0^{n-1}\in\mathcal{Y}^n} W^n(y_0^{n-1}\mid \phi(a))\mathbb{I}\{\psi(y_0^{n-1})\ne a\}.
\end{equation*}

In order to measure efficiency of communication,
\textit{coding rate}, defined as $\log M/n$,
is considered.
Shannon and other researchers showed that there exists the asymptotically best trade-off between coding rate and error probability of code.

\begin{theorem}[Channel coding theorem]
There exists a quantity $C(W)\in(0,1)$, called capacity of a channel $W$, which has the following properties.

There exists sequences of encoders $\phi_i: \mathcal{M}_i\to X^{n_i}$
and decoders $\psi_i: Y^{n_i}\to\mathcal{M}_i$ such that error probabilities tend to 0 
and limit superior of $\log |\mathcal{M}_i|/n_i$ is smaller than $C(W)$.

Conversely, 
for any sequences of encoders $\phi_i: \mathcal{M}_i\to X^{n_i}$ and decoders $\psi_i: Y^{n_i}\to\mathcal{M}_i$,
where limit inferior of $\log |\mathcal{M}_i|/n_i$ is larger than $C(W)$, error probabilities tend to 1.
\end{theorem}

The channel coding theorem only shows existence of sequences of encoders and decoders on which reliable efficient communication is possible.
One of the goals of coding theory is to find practical encoders and decoders
which achieve the best trade-off described in the channel coding theorem.

\section{Preview of Polar Codes}
Polar codes, introduced by Ar{\i}kan~\cite{5075875}, are the first provably capacity achieving codes for
any symmetric binary-input discrete memoryless channels (B-DMC) which have low complexity encoding and decoding algorithms.
Complexities of encoding and decoding are both $O(N\log N)$ where $N$ is the blocklength.
Polar codes are based on \textit{channel polarization} phenomenon.

Ar{\i}kan and Telatar showed that asymptotic error probability of polar codes whose coding rate is smaller than
capacity is $o(2^{-N^\beta})$ for any $\beta<1/2$ and $\omega(2^{-N^\beta})$ for any $\beta>1/2$~\cite{arikan2008rcp}.
Since error probabilities of the best codes decay exponentially in the blocklength~\cite{gallager1968information},
polar codes are not optimal in the asymptotic region.

In the original work of Ar{\i}kan, generator matrices of polar codes are constructed by
choosing rows of $G^{\otimes n}$, where $G=\begin{bmatrix}1&0\\1&1\end{bmatrix}$ and where ${}^{\otimes n}$ denotes the Kronecker power.
On the other hand, Korada, {\c S}a{\c s}o{\u g}lu, and Urbanke generalized polar codes
which are constructed from larger matrices instead of $G$~\cite{korada2009pcc}.
Further, they showed that asymptotic performance of polar codes is improved by using larger matrices.

Korada and Urbanke showed that polar codes also achieve symmetric rate-distortion trade-off as lossy source codes~\cite{korada2009pco}.
They also showed that polar codes achieve optimal rate of Wyner-Ziv and Gelfand-Pinsker problems.

\section{Contribution of the Thesis}
\subsection{Construction of polar codes}
In Ar{\i}kan's original work, complexity of construction of polar codes grows exponentially in the blocklength.
We show a novel construction method whose complexity is linear in the blocklength.
The construction method is based on \textit{density evolution}, which has been used for
calculation of the large blocklength limit of the bit error probability of low-density parity-check (LDPC) codes~\cite{RiU05/LTHC}.

\subsection{Generalization of polar codes}
Non-binary polar codes are considered.
When a set of input alphabets is a finite field, we obtain sufficient conditions for a matrix
on which capacity-achieving polar codes can be constructed for any DMC\@.
We also consider polar codes constructed from a non-linear mapping instead of a linear mapping.

\section{Organization of the Thesis}
In Chapter~\ref{chap:bpolar}, channel polarization phenomenon for B-DMC, introduced by Ar{\i}kan~\cite{5075875}, is considered.
In Chapter~\ref{chap:speed}, the speed of channel polarization, shown by Ar{\i}kan and Telatar~\cite{arikan2008rcp}, is considered.
In Chapter~\ref{chap:pcodes}, we define polar codes which are based on the channel polarization phenomenon~\cite{5075875}.
It is shown that complexities of encoding and decoding are
$O(N\log N)$ where $N$ is the blocklength.
We show a novel construction method whose complexity is linear in the blocklength~\cite{5205857} for symmetric B-DMC\@.
In Chapter~\ref{chap:qary}, channel polarization of $q$-ary channels is considered.
Sufficient conditions for channel polarization matrices and a simple example are shown.

\section{Notations and Useful Facts}
In the thesis, we use the following notations.
Let $x_0^{n-1}$ and $x_i^j$ denote a row vector $(x_0,\dots,x_{n-1})$ and its subvector $(x_i,\dots,x_j)$.
For $\mathcal{A}=(a_0,\dotsc,a_{m-1})\subseteq \{0,\dotsc,n-1\}$, $x_\mathcal{A}$ denotes a subvector $(x_{a_0},\dotsc,x_{a_{m-1}})$.
Let $\mathcal{F}^c$ denote the complement of a set $\mathcal{F}$, and $|\mathcal{F}|$ denote cardinality of $\mathcal{F}$.
Let $G_{ij}$ denote $(i,j)$ element of a matrix $G$.

Let $X$, $Y$ and $Z$ be random variables on a probability space $(\Omega,\, \mathcal{F},\, P)$
ranging on discrete sets $\mathcal{A}$, $\mathcal{B}$ and $\mathcal{C}$, respectively.
The mutual information between $X$ and $Y$ is defined as
\begin{equation*}
I(X; Y) := \sum_{x\in\mathcal{A}, y\in\mathcal{B}}P(X=x, Y=y) \log \frac{P(X=x, Y=y)}{P(X=x)P(Y=y)}.
\end{equation*}
Similarly, the mutual information between $X$ and $(Y,Z)$ is defined as
\begin{equation*}
I(X; YZ) := \sum_{x\in\mathcal{A}, y\in\mathcal{B}, z\in\mathcal{C}}P(X=x, Y=y, Z=z) \log \frac{P(X=x, Y=y, Z=z)}{P(X=x)P(Y=y, Z=z)}.
\end{equation*}
The conditional mutual information between $X$ and $Y$ given $Z$ is defined as
\begin{equation*}
I(X; Y\mid Z) := \sum_{x\in\mathcal{A}, y\in\mathcal{B}, z\in\mathcal{C}}P(X=x, Y=y, Z=z) \log
\frac{P(X=x, Y=y\mid Z=z)}{P(X=x\mid Z=z)P(Y=y\mid Z=z)}.
\end{equation*}
The most fundamental fact in the thesis, called \textit{the chain rule for mutual information}, is the following.
\begin{proposition}\label{prop:chain}\cite{gallager1968information}
\begin{equation*}
I(X;YZ) = I(X;Y) + I(X;Z\mid Y)
\end{equation*}
\end{proposition}
The cutoff rate of $(X, Y)$ is defined as
\begin{equation*}
R_0(X; Y) := -\log \sum_{y\in\mathcal{B}}\left[\sum_{x\in\mathcal{A}} P(X=x)\sqrt{P(Y=y\mid X=x)}\right]^2.
\end{equation*}
Similarly, the conditional cutoff rate of $(X, Y)$ given $Z$ is defined as
\begin{equation*}
R_0(X; Y\mid Z) := -\log \sum_{y\in\mathcal{B},z\in\mathcal{C}}P(Z=z)
\left[\sum_{x\in\mathcal{A}} P(X=x\mid Z=z)\sqrt{P(Y=y\mid X=x, Z=z)}\right]^2.
\end{equation*}
In the thesis, the cutoff rate is used for bounding the mutual information by the following proposition.
\begin{proposition}\label{prop:cutoff}\cite{gallager1968information}
\begin{align*}
I(X; Y) &\ge R_0(X; Y)\\
I(X; Y\mid Z) &\ge R_0(X; Y\mid Z)
\end{align*}
\end{proposition}
\begin{proof}
The second inequality is an immediate consequence of the first inequality.
\begin{align*}
I(X; Y) &= \sum_{x\in\mathcal{A}, y\in\mathcal{B}}P(X=x, Y=y) \log \frac{P(X=x, Y=y)}{P(X=x)P(Y=y)}\\
&= -2\sum_{x\in\mathcal{A}, y\in\mathcal{B}}P(X=x, Y=y) \log \sqrt{\frac{P(X=x)P(Y=y)}{P(X=x, Y=y)}}\\
&\ge -2\sum_{y\in\mathcal{B}}P(Y=y) \log \sum_{x\in\mathcal{A}} P(X=x\mid Y=y)\sqrt{\frac{P(X=x)P(Y=y)}{P(X=x, Y=y)}}\\
&= -\sum_{y\in\mathcal{B}}P(Y=y) \log \left[\sum_{x\in\mathcal{A}} P(X=x\mid Y=y)\sqrt{\frac{P(X=x)P(Y=y)}{P(X=x, Y=y)}}\right]^2\\
&\ge - \log \sum_{y\in\mathcal{B}}P(Y=y)\left[\sum_{x\in\mathcal{A}} P(X=x\mid Y=y)\sqrt{\frac{P(X=x)P(Y=y)}{P(X=x, Y=y)}}\right]^2\\
&= -\log \sum_{y\in\mathcal{B}}\left[\sum_{x\in\mathcal{A}} P(X=x)\sqrt{P(Y=y\mid X=x)}\right]^2 = R_0(X;Y)
\end{align*}
The above inequalities are obtained from Jensen's inequality.
\end{proof}

\chapter{Channel Polarization of B-DMCs by Linear Kernel}\label{chap:bpolar}
\section{Introduction}
Ar{\i}kan introduced polar codes whose generator matrix is constructed
by choosing rows from $\begin{bmatrix}1&0\\1&1\end{bmatrix}^{\otimes n}$~\cite{5075875}.
Korada, {\c S}a{\c s}o{\u g}lu, and Urbanke generalized the result for an arbitrary full-rank matrix~\cite{korada2009pcc}.
Ar{\i}kan explained that polar codes are constructed on \textit{channel polarization} phenomenon.
This explanation is useful for understanding polar codes.
In this chapter, we consider the channel polarization phenomenon of B-DMC induced by an arbitrary linear mapping.

\section{Preliminaries}\label{sec:prel}
Let $\mathcal{X}$ and $\mathcal{Y}$ be sets of input alphabets and output alphabets.
In the thesis, we assume that $\mathcal{X}$ is a finite set and $\mathcal{Y}$ is at most a countable set.
A DMC is defined as a conditional probability distribution $W(y\mid x)$ over $\mathcal{Y}$ for all $x\in\mathcal{X}$.
We write $W:\mathcal{X}\to\mathcal{Y}$ to mean a DMC with
a set of input alphabets $\mathcal{X}$ and a set of output alphabets $\mathcal{Y}$.
In this chapter, we deal with B-DMC, i.e., $\mathcal{X}=\{0,1\}$
and assume that the base of logarithm is 2.
\begin{definition}
The symmetric capacity of a B-DMC $W:\mathcal{X}\to\mathcal{Y}$ is defined as
\begin{equation*}
I(W) := \sum_{x\in\mathcal{X}}\sum_{y\in\mathcal{Y}} \frac12 W(y\mid x)
\log\frac{W(y\mid x)}{\frac12 W(y\mid 0) + \frac12 W(y\mid 1)}.
\end{equation*}
Note that $I(W)\in[0,1]$.
\end{definition}

\begin{definition}
The Bhattacharyya parameter of a B-DMC $W$ is defined as
\begin{equation*}
Z(W) := \sum_{y\in\mathcal{Y}} \sqrt{W(y\mid 0)W(y\mid 1)}.
\end{equation*}
Note that $Z(W) \in [0,1]$.
\end{definition}
\begin{lemma}\cite{5075875}\label{lem:bIZ}
The symmetric capacity and the Bhattacharyya parameter satisfy the following relations.
\begin{align*}
I(W)+Z(W)&\ge 1\\
I(W)^2+Z(W)^2&\le 1
\end{align*}
\end{lemma}

\section{Channel Polarization}\label{sec:polarization}
We consider recursive channel transform using a full-rank square matrix $G$ on $\mathbb{F}_2$.
In~\cite{5075875}, Ar{\i}kan chose 
\begin{equation}
G=\begin{bmatrix}1&0\\1&1\end{bmatrix}.\label{eq:2x2}
\end{equation}
In this chapter, following Korada, {\c S}a{\c s}o{\u g}lu, and Urbanke~\cite{korada2009pcc},
we assume that $G$ is an arbitrary full-rank square matrix.
Let $\ell$ be the size of $G$.
Channel transform procedure is defined as follows.
\begin{definition}
\begin{align*}
W^\ell(y_0^{\ell-1}\mid x_0^{\ell-1}) &:=\prod_{i=0}^{\ell-1} W(y_i\mid x_i)\\
W^{(i)}(y_0^{\ell-1},u_0^{i-1}\mid u_i) &:= \frac1{2^{\ell-1}}
\sum_{u_{i+1}^{\ell-1}} W^{\ell}(y_0^{\ell-1}\mid u_0^{\ell-1}G).
\end{align*}
\end{definition}
\noindent
In the above definition, $W^{(i)}$ is called a \textit{subchannel} of $W$.
Let $U_0^{\ell-1}$, $X_0^{\ell-1}$ and $Y_0^{\ell-1}$ denote random variables taking values on
$\mathcal{X}^\ell$, $\mathcal{X}^\ell$ and $\mathcal{Y}^\ell$, respectively,
and obeying distribution
\begin{equation*}
P(U_0^{\ell-1} = u_0^{\ell-1},~X_0^{\ell-1} = x_0^{\ell-1},~Y_0^{\ell-1} = y_0^{\ell-1}) = \frac1{2^\ell} W^\ell(y_0^{\ell-1}\mid u_0^{\ell-1} G)
\mathbb{I}\left\{x_0^{\ell-1} V = u_0^{\ell-1} \right\}
\end{equation*}
where $V$ is an $\ell\times\ell$ full-rank upper triangle matrix.
Since there exists a one-to-one correspondence between $U_0^i$ and $X_0^i$ for all $i\in\{0,\dotsc,\ell-1\}$,
statistical properties of $W^{(i)}$ are invariant under an operation $G\to VG$.
Further, a permutation of columns of $G$ does not change statistical properties of $W^{(i)}$.
Since any full-rank matrix can be decomposed as $VLP$ where $V$, $L$, and $P$ are
upper triangle, lower triangle, and permutation matrices, respectively,
without loss of generality we assume that $G$ is a lower triangle matrix.

Assume that $\{B_i\}_{i\in\mathbb{N}}$ is a sequence of independent uniform random variables taking values on $\{0,\dots,\ell-1\}$.
Let $I_n := I(W^{(B_1)\dotsm(B_n)})$.
Channel polarization phenomenon is described in the following theorem.
\begin{theorem}\cite{5075875}, \cite{korada2009pcc}\label{thm:polarization}
If $G$ is not diagonal,
$I_n\to I_\infty$ almost surely, where $I_\infty$ satisfies
\begin{equation*}
I_\infty = \begin{cases}
0, &\text{\rm with probability } 1-I(W)\\
1, &\text{\rm with probability } I(W).
\end{cases}
\end{equation*}
\end{theorem}
\noindent
Theorem~\ref{thm:polarization} says that $\ell^n$ subchannels $\{W^{(b_1)\dotsm(b_n)}\}_{(b_1,\dotsc,b_n)\in\{0,\dotsc,\ell-1\}^n}$
are polarized between noiseless channels and pure noisy channels for sufficiently large $n$.
The first part of Theorem~\ref{thm:polarization} is proven by the martingale convergence theorem
without using the assumption that $G$ is not diagonal.

\begin{lemma}\label{lem:asc}
$\lim_{n\to\infty} I_n$ exists almost surely.
\end{lemma}
\begin{proof}
Let $U_0^{\ell-1}$ and $Y_0^{\ell-1}$ denote random variables taking values on $\mathcal{X}^\ell$ and $\mathcal{Y}^\ell$, respectively,
and obeying the distribution
\begin{equation*} 
P(U_0^{\ell-1} = u_0^{\ell-1}, Y_0^{\ell-1} = y_0^{\ell-1}) = \frac1{2^\ell} W^\ell(y_0^{\ell-1}\mid u_0^{\ell-1} G).
\end{equation*}
From the chain rule for mutual information, shown in Proposition~\ref{prop:chain}, one obtains
\begin{equation*}
\ell I(W) = I(U_0^{\ell-1}; Y_0^{\ell-1}) = \sum_{i=0}^{\ell-1} I(U_i; Y_0^{\ell-1} \mid U_0^{i-1})
= \sum_{i=0}^{\ell-1} I(U_i; Y_0^{\ell-1}, U_0^{i-1})
= \sum_{i=0}^{\ell-1} I(W^{(i)}).
\end{equation*}
Hence, $I_n$ is a bounded martingale.
From the martingale convergence theorem,
$\lim_{n\to\infty} I_n$ exists almost surely~\cite{billingsley1995probability}.
\end{proof}

\begin{proof}[proof of Theorem~\ref{thm:polarization}]
Let $k$ denote the largest number where Hamming weight of $k$-th row of $G$ is larger than 1.
Hence, 
\begin{equation*}
 W^{(k)}(y_0^{\ell-1}, u_0^{k-1}\mid u_k)
= \frac1{2^{\ell-1}} \prod_{j\in S_0} W(y_j\mid x_j)\prod_{j\in S_1} W(y_j\mid u_k + x_j)
\prod_{j=k+1}^{\ell-1} \left(W(y_j\mid 0) + W(y_j\mid 1)\right)
\end{equation*}
where $S_0 := \{i \in \{0,\dotsc,k\}\mid G_{ki} = 0\}$,
$S_1 := \{i \in \{0,\dotsc,k\}\mid G_{ki} = 1\}$,
 and $x_j$ is $j$-th element of $(u_0^{k-1},0_k^{\ell-1})G$.
Let 
\begin{equation*}
W^{(k)'}(y_i, y_k\mid u_k) := W(y_i \mid u_k) W(y_k \mid u_k) 
\end{equation*}
where $i\in S_0$.

From Lemma~\ref{lem:asc},
\begin{equation*}
\lim_{n\to\infty} |I(W_{n+1})-I(W_n)|=0, \hspace{2em} \text{with probability 1.}
\end{equation*}
Hence,
\begin{equation}\label{eq:ctz}
\lim_{n\to\infty} I(W_n^{(k)'})-I(W_n)=0, \hspace{2em} \text{with probability 1.}
\end{equation}
Let $(\Omega = \mathcal{X}\times\mathcal{Y}^2,~2^\Omega,~P)$ denote a probability space where
\begin{equation*}
P((u,y_1,y_2)) := \frac12 W_n(y_1 \mid u) W_n(y_2 \mid u) 
\end{equation*}
for $(u,y_1,y_2)\in\Omega$, and $(U, Y_1, Y_2)$ denote random variables obeying the distribution $P$.
From~\eqref{eq:ctz},
$I(Y_1, Y_2; U) - I(Y_1; U) = I(Y_2; U\mid Y_1)\to 0$ for all $x\in\mathcal{X}$.
Since mutual information is lower bounded by cutoff rate as shown in Proposition~\ref{prop:cutoff}, one obtains
\begin{align*}
I(Y_2; U\mid Y_1) &\ge -\log\sum_{y_1\in\mathcal{Y}_n,y_2\in\mathcal{Y}_n}P(Y_1=y_1)
\left(\sum_{u\in\mathcal{X}} P(U=u\mid Y_1=y_1)\sqrt{P(Y_2=y_2\mid U=u, Y_1=y_1)}\right)^2\\
&=-\log \sum_{y_1\in\mathcal{Y}_n} P(Y_1=y_1)\left[1-2P(U=0\mid Y_1=y_1)P(U=1\mid Y_1=y_1)(1-Z(W_n))\right]\\
&=-\log \left[1-2\sum_{y_1\in\mathcal{Y}_n} P(Y_1=y_1)\left(\sqrt{P(U=0\mid Y_1=y_1)P(U=1\mid Y_1=y_1)}\right)^2(1-Z(W_n))\right]\\
&\ge-\log \left[1-2\left(\sum_{y_1\in\mathcal{Y}_n} P(Y_1=y_1)\sqrt{P(U=0\mid Y_1=y_1)P(U=1\mid Y_1=y_1)}\right)^2(1-Z(W_n))\right]\\
&=-\log \left[1-2\left(\frac12 Z(W_n)\right)^2(1-Z(W_n))\right]\\
&=-\log \left[1-\frac1{2}Z(W_n)^2(1-Z(W_n))\right]\\
&\ge-\log \left[1-\frac1{2}(1-I(W_n))^2\left(1-\sqrt{1-I(W_n)^2}\right)\right].
\end{align*}
The last inequality is obtained from Lemma~\ref{lem:bIZ}.
Since the left-hand side of the above inequality tends to 0 with probability 1,
we conclude $I_\infty \in \{0,1\}$ with probability $1$.
Since $I_n$ is a martingale, $I_\infty = 1$ with probability $I(W)$.
\end{proof}

\chapter{Speed of Polarization}\label{chap:speed}
\section{Introduction}
In this chapter, we consider how fast $W_n$ are polarized between noiseless channel and pure noisy channel.
Instead of $I(W_n)$, we evaluate Bhattacharyya parameter $Z(W_n)$ which has the relation with $I(W_n)$ as shown in Lemma~\ref{lem:bIZ}.
Let $Z_n := Z(W^{(B_1)\dotsm(B_n)})$. From Theorem~\ref{thm:polarization} and Lemma~\ref{lem:bIZ},
$Z_n\to Z_\infty$ almost surely where $Z_\infty$ satisfies $Z_\infty=0$ with probability $I(W)$, and $Z_\infty=1$ with probability $1-I(W)$.
Hence, for any $\epsilon\in(0,1)$
\begin{equation*}
\lim_{n\to\infty} P(Z_n < \epsilon) = I(W).
\end{equation*}
Ar{\i}kan and Telatar showed a stronger result when $G$ is the $2\times 2$ matrix~\eqref{eq:2x2} as follows~\cite{arikan2008rcp}, \cite{5205856}.
\begin{proposition}
For any $\beta<1/2$,
\begin{equation*}
\lim_{n\to\infty} P(Z_n < 2^{-2^{\beta n}}) = I(W).
\end{equation*}
For any $\beta>1/2$,
\begin{equation*}
\lim_{n\to\infty} P(Z_n < 2^{-2^{\beta n}}) = 0.
\end{equation*}
\end{proposition}
\noindent
Korada, {\c S}a{\c s}o{\u g}lu and Urbanke generalized the above result to general matrices~\cite{korada2009pcc}.
Further, Tanaka and Mori showed a more detailed speed of polarization~\cite{tanaka2010rre}.

\section{Preliminaries}
\begin{definition}
Partial distance $D^{[i]}$ of $G$ is defined as
\begin{equation*}
D^{[i]}:=\min_{v_{i+1}^{\ell-1}, w_{i+1}^{\ell-1}} d((0_0^{i-1}, 0, v_{i+1}^{\ell-1})G,\, (0_0^{i-1}, 1, w_{i+1}^{\ell-1})G)
\end{equation*}
where $d(a, b)$ denotes the Hamming distance between $a\in\mathcal{X}^\ell$ and $b\in\mathcal{X}^\ell$,
and where $0_0^{i-1}$ denotes the all-zero vector of length $i$.
\end{definition}
\noindent
Partial distance plays a central role in evaluation of speed of the polarization phenomenon.
\begin{lemma}\cite{korada2009pcc}\label{lem:bZD}
\begin{equation*}
Z(W)^{D^{[i]}} \le Z(W^{(i)})\le 2^{\ell-i}Z(W)^{D^{[i]}}.
\end{equation*}
\end{lemma}
\begin{definition}\label{def:exp}
The exponent of a matrix $G$ is defined as $E(G) := (1/\ell)\sum_{i=0}^{\ell-1}\log_\ell D^{[i]}$.
The second exponent of a matrix $G$ is defined as $V(G) := (1/\ell)\sum_{i=0}^{\ell-1}(\log_\ell D^{[i]}-E(G))^2$.
\end{definition}
\begin{definition}
The Q function is defined as
\begin{equation*}
Q(t) := \frac1{\sqrt{2\pi}}\int_t^\infty \exp\left\{-\frac{x^2}{2}\right\}\mathrm{d}x.
\end{equation*}
\end{definition}
In this chapter, the base of logarithm is assumed to be 2 unless otherwise stated.

\section{Speed of Polarization}
\subsection{Speed of polarization and random process}
The following result is obtained by Ar{\i}kan and Telatar~\cite{arikan2008rcp}
when $G$ is the $2\times 2$ matrix~\eqref{eq:2x2},
and by Korada, {\c S}a{\c s}o{\u g}lu, and Urbanke for the general case~\cite{korada2009pcc}.
\begin{theorem}\label{thm:rop}
\begin{equation*}
\lim_{n\to\infty} P(Z_n < 2^{-\ell^{\beta n}}) = I(W)
\end{equation*}
for any $\beta<E(G)$.
\begin{equation*}
\lim_{n\to\infty} P(Z_n < 2^{-\ell^{\beta n}}) = 0
\end{equation*}
for any $\beta>E(G)$.
\end{theorem}
\noindent
Tanaka and Mori showed more detailed speed of polarization~\cite{tanaka2010rre}.
\begin{theorem}\label{thm:rdr}
For any $f(n)=o(\sqrt{n})$,
\begin{equation*}
\lim_{n\to\infty} P\left(Z_n < 2^{-\ell^{E(G) n + t\sqrt{V(G)n}+f(n)}}\right) = I(W)Q(t).
\end{equation*}
\end{theorem}


In order to prove Theorem~\ref{thm:rop} and Theorem~\ref{thm:rdr},
we consider a generalized process.
Let $\{S_n\}_{n\in\mathbb{N}}$ be independent and identically distributed random variables ranging on $[1,\infty)$.
Assume that the expectation and the variance of $\log S_1$ exist,
and are denoted by $\mathbb{E}[\log S_1]$ and $\mathbb{V}[\log S_1]$, respectively.
The random process $\{Z_n\in(0,1)\}_{n\in\mathbb{N}}$ satisfies the following conditions.
\begin{itemize}
\item[(c1)] $Z_n\to Z_\infty$ almost surely.
\item[(c2)] There exists a positive constant $c_0$ such that $c_0 Z_n^{S_n}\le Z_{n+1}$.
\item[(c3)] There exists a positive constant $c_1$ such that $Z_{n+1}\le c_1 Z_n^{S_n}$.
\item[(c4)] $S_n$ is independent of $Z_m$ for $m\le n$.
\end{itemize}
In the following proof, the above conditions are used.
The random process $\{Z_n\}_{n\in\mathbb{N}}$ satisfies (c2) and (c3) when $S_n = D^{[B_n]}$.
Then, it holds that $\mathbb{E}[\log S_1] = E(G)\log\ell$ and that $\mathbb{V}[\log S_1] = V(G)(\log\ell)^2$.
Let $\mathcal{T}_m^n(\gamma) := \{\omega\in\Omega\mid Z_k(\omega) < \gamma, \forall k \in \{m, m+1,\dotsc,n\}\}$ and
$\mathcal{T}_m^\infty(\gamma) := \cap_{n=1}^\infty \mathcal{T}_m^n(\gamma)$.
From (c1), there exist zero sets $\mathcal{A}$ and $\mathcal{B}$ where $P(\mathcal{A})=P(\mathcal{B})=0$ such that
\begin{equation*}
\{\omega\in\Omega\mid Z_\infty(\omega) < \gamma\} \subseteq \left(\bigcup_{k=1}^\infty \mathcal{T}_k^\infty(\gamma)\right) \cup\mathcal{A}
\subseteq \{\omega\in\Omega\mid Z_\infty(\omega)\le \gamma\} \cup\mathcal{B}
\end{equation*}
for any $\gamma\in[0,1]$.


\subsection{Direct part of Theorem~\ref{thm:rop}}
\begin{proposition}\label{prop:direct}
Let $\{X_n\}_{n\in\mathbb{N}}$ be a random process satisfying (c1) and (c3).
For any fixed $\beta\in(0,\mathbb{E}[\log S_1])$
\begin{equation*}
\lim_{n\to\infty} P\left(X_n\le 2^{-2^{\beta n}}\right) = P(X_\infty=0).
\end{equation*}
\end{proposition}
\begin{proof}
Fix $\epsilon \in(0,1)$.
We consider a process $\{L_i\}$ defined on the basis of $\{X_i\}$ as
\begin{align*}
L_i&=\log\log (1/X_i),\hspace{3em} i=0,\dotsc,m\\
L_{i+1}&= \log (S_{i}-\epsilon)+L_{i},\hspace{2em} i > m.
\end{align*}
Fix $\zeta>\max\{1,c_1\}$.
Conditional on $\mathcal{T}_m^{m+k-1}(\zeta^{-1/\epsilon})$, the inequality $\log\log(1/X_n) \ge L_n$ holds for any $n\in\{m,m+1,\dotsc,m+k\}$.
On the other hand, it holds
\begin{equation*}
L_{m+k} = L_m + \sum_{i=m}^{m+k-1} \log(S_i-\epsilon)
\ge L_m + \sum_{i=m}^{m+k-1} (\log S_i + \log (1-\epsilon)).
\end{equation*}
Conditional on $\mathcal{C}_m^{m+k-1}:=\{(1/k)\sum_{i=m}^{m+k-1}\log S_i \ge \mathbb{E}[\log S_1]-\epsilon\}$,
it holds
\begin{equation*}
L_{m+k} \ge k(\mathbb{E}[\log S_1] - \epsilon + \log (1-\epsilon)) + L_m.
\end{equation*}
Hence, 
\begin{align*}
P\left(\log\log(1/X_{m+k}) \ge k(\mathbb{E}[\log S_1] - \epsilon + \log(1-\epsilon)) + L_m\right)
&\ge P\left(\mathcal{T}_m^{m+k-1}(\zeta^{-1/\epsilon}) \cap \mathcal{C}_m^{m+k-1}\right)\\
&\ge 1-P\left(\mathcal{T}_m^{m+k-1}(\zeta^{-1/\epsilon})^c\right)-P\left({\mathcal{C}_m^{m+k-1}}^c\right).
\end{align*}
From the law of large numbers, it holds
$\lim_{k\to\infty} P\left({\mathcal{C}_m^{m+k-1}}^c\right) = 0$.
Since $X_n$ converges to $X_\infty$ almost surely,
$\lim_{m\to\infty}P(\mathcal{T}_m^\infty(\zeta^{-1/\epsilon})) \ge P(X_\infty<\zeta^{-1/\epsilon})$.
On the other hand,
we observe
\begin{multline*}
\liminf_{k\to\infty}P(\log\log(1/X_{m+k}) 
\ge k(\mathbb{E}[\log S_1] - \epsilon+\log(1-\epsilon)) + L_m)\\
\le\liminf_{n\to\infty} P\Biggl(\frac1n\log\log(1/X_n)
\ge \mathbb{E}[\log S_1] - \gamma\Biggr)
\end{multline*}
for any $\gamma > \epsilon - \log (1-\epsilon)$.
Hence,
\begin{equation*}
\liminf_{n\to\infty} P\Biggl(\frac1n\log\log(1/X_n)
\ge \mathbb{E}[\log S_1] - \gamma\Biggr)
\ge P(X_\infty<\zeta^{-1/\epsilon}).
\end{equation*}
\end{proof}

\subsection{Converse part of Theorem~\ref{thm:rop}}
\begin{proposition}\label{prop:conv}
Let $\{X_n\}_{n\in\mathbb{N}}$ be a random process satisfying (c1) and (c2).
For any fixed $\beta>\mathbb{E}[\log S_1]$
\begin{equation*}
\lim_{n\to\infty} P\left(X_n\le 2^{-2^{\beta n}}\right) = 0.
\end{equation*}
\end{proposition}
\begin{proof}
Fix $\epsilon \in(0,1)$.
We consider a process $\{L_i\}$ defined on the basis of $\{X_i\}$ as
\begin{align*}
L_i&=\log\log (1/X_i), \hspace{2em} i=0,\dotsc,m\\
L_{i+1}&= \log (S_{i}+\epsilon)+L_{i}, \hspace{1em} i > m.
\end{align*}
Fix $\zeta\in(0,\min\{c_0, 1\})$.
Conditional on $\mathcal{T}_m^\infty(\zeta^{1/\epsilon})$, it holds $\log\log(1/X_n) \le L_n$ for any $n\ge m$.
It holds
\begin{equation*}
L_{m+k} = L_m + \sum_{i=m}^{m+k-1} \log(S_i+\epsilon)
\le L_m + \sum_{i=m}^{m+k-1} (\log S_i +\epsilon).
\end{equation*}
For any $\gamma>0$,
\begin{align*}
&\limsup_{n\to\infty} P\left(\frac1n\log\log(1/X_n)\ge \mathbb{E}[\log S_1] + 2\epsilon\right)\\
&=\limsup_{k\to\infty} P\left(\frac1{m+k}\log\log(1/X_{m+k})\ge \mathbb{E}[\log S_1] + 2\epsilon\right)\\
&\le\limsup_{k\to\infty} \left\{
P\left(\frac1{m+k}L_{m+k}\ge \mathbb{E}[\log S_1] + 2\epsilon \;\bigcap\; \mathcal{T}_m^\infty(\zeta^{1/\epsilon})\right)
+P\left(X_{m+k}\le \gamma \;\bigcap\; \mathcal{T}_m^\infty(\zeta^{1/\epsilon})^c\right)\right\}\\
&\le\limsup_{k\to\infty} \left\{
P\left(\frac1{m+k}L_{m+k}\ge \mathbb{E}[\log S_1] + 2\epsilon\right)
+P\left(X_{m+k}\le \gamma \;\bigcap\; \mathcal{T}_m^\infty(\zeta^{1/\epsilon})^c\right)\right\}\\
&\le\limsup_{k\to\infty} \Biggl\{
P\Biggl(\frac1{m+k}\left(L_m+k\epsilon+\sum_{i=m}^{m+k-1}\log S_i\right)
\ge \mathbb{E}[\log S_1] + 2\epsilon\Biggr)\Biggr\}\\
&\quad+P\left(X_\infty\le \gamma \;\bigcap\; \mathcal{T}_m^\infty\left(\zeta^{1/\epsilon}\right)^c\right)\\
&=
P\left(X_\infty\le \gamma \;\bigcap\; \mathcal{T}_m^\infty\left(\zeta^{1/\epsilon}\right)^c\right)
\end{align*}
The last equality is obtained from the law of large numbers.
\begin{align*}
\lim_{m\to\infty}P\left(X_\infty\le \gamma \;\bigcap\; \mathcal{T}_m^\infty(\zeta^{1/\epsilon})^c\right)
&= 1- \lim_{m\to\infty}P\left(X_\infty> \gamma \;\bigcup\; \mathcal{T}_m^\infty(\zeta^{1/\epsilon})\right)\\
&\le 1- P\left(X_\infty> \gamma \;\bigcup\; X_\infty < \zeta^{1/\epsilon}\right)
\end{align*}
By letting $\gamma=\zeta^{1/\epsilon}/2$, the right-hand side of the above inequality is equal to zero.
\end{proof}

\subsection{Direct part of Theorem~\ref{thm:rdr}}
\begin{proposition}\label{prop:rdirect}
Let $\{X_n\}_{n\in\mathbb{N}}$ be a random process satisfying (c1), (c3) and (c4).
For any $f(n)=o(\sqrt{n})$,
\begin{equation*}
\liminf_{n\to\infty} P\left(X_n < 2^{-2^{\mathbb{E}[\log S_1]n+t\sqrt{\mathbb{V}[\log S_1]n}+f(n)}}\right) \ge P(X_\infty=0)Q(t).
\end{equation*}
\end{proposition}
\begin{proof}
Let $L_n := \log X_n$.
Let $\gamma := \max\{2, c_1\}$.
One obtains
\begin{equation*}
L_n \le \log \gamma + S_{n-1} L_{n-1}
\le \left(\sum_{j=m}^{n-1}\prod_{i=j+1}^{n-1} S_i\right)\log \gamma + \left(\prod_{i=m}^{n-1}S_i\right) L_m
\le \left(\prod_{i=m}^{n-1} S_i\right)\left((n-m)\log \gamma + L_m\right).
\end{equation*}
Fix $\beta\in(0,E(G))$.
Let $m:= (\log n + \log\log\gamma)/\beta$.
Conditioned on $\mathcal{D}_m(\beta) := \{\omega\in\Omega\mid X_m(\omega) < 2^{-2^{\beta m}}\}$,
\begin{equation*}
L_n \le -\left(\prod_{i=m}^{n-1} S_i\right)m\log\gamma.
\end{equation*}
Let $\mathcal{H}_{m}^{n-1}(t) := \{\sum_{i=m}^{n-1} \log S_i \ge (n-m)\mathbb{E}[\log S_1] + t\sqrt{\mathbb{V}[\log S_1] (n-m)} + f(n-m)\}$
where $f(k) = o(\sqrt{k})$.
Conditioned on $\mathcal{D}_m(\beta)$ and $\mathcal{H}_{m}^{n-1}(t)$,
it holds
\begin{equation*}
\log (-L_n) \ge \log m + \log\log\gamma + 
(n-m)\mathbb{E}[\log S_1] + t\sqrt{\mathbb{V}[\log S_1] (n-m)} + f(n-m).
\end{equation*}
Hence, it holds
\begin{multline*}
P\left(\log (-L_n) \ge \log m + \log\log\gamma + 
(n-m)\mathbb{E}[\log S_1] + t\sqrt{\mathbb{V}[\log S_1] (n-m)} + f(n-m)\right)\\
\ge P\left(\mathcal{D}_m(\beta)\cap\mathcal{H}_{m}^{n-1}(t)\right)
= P\left(\mathcal{D}_m(\beta)\right)P\left(\mathcal{H}_{m}^{n-1}(t)\right).
\end{multline*}
The last equality follows from (c4).
From Theorem~\ref{thm:rop}, it holds $\lim_{m\to\infty}P\left(\mathcal{D}_m(\beta)\right)=P(X_\infty=0)$.
From the central limit theorem, it holds $\lim_{n\to\infty}P\left(\mathcal{H}_{m}^{n-1}(t)\right)=Q(t)$.
At last, one obtains
\begin{equation*}
\liminf_{n\to\infty} P\left(\log\log(1/X_n)  \ge
n\mathbb{E}[\log S_1] + t\sqrt{\mathbb{V}[\log S_1] n} + f(n)\right)
\ge P(X_\infty = 0)Q(t).
\end{equation*}
for any $f(n)=o(\sqrt{n})$.
\end{proof}

\subsection{Converse part of Theorem~\ref{thm:rdr}}
\begin{proposition}\label{prop:rconv}
Let $\{X_n\}_{n\in\mathbb{N}}$ be a random process satisfying (c1), (c2) and (c4).
For any $f(n)=o(\sqrt{n})$,
\begin{equation*}
\limsup_{n\to\infty} P\left(X_n < 2^{-2^{\mathbb{E}[\log S_1]n+t\sqrt{\mathbb{V}[\log S_1]n}+f(n)}}\right) \le P(X_\infty=0)Q(t).
\end{equation*}
\end{proposition}
\begin{proof}
Let $L_n := \log X_n$.
Let $\gamma := \min\{1, c_0\}$.
For any $m\le n$, one obtains
\begin{equation*}
L_n \ge \log \gamma + S_{n-1} L_{n-1}
\ge \left(\sum_{j=m}^{n-1}\prod_{i=j+1}^n S_i\right)\log \gamma + \left(\prod_{i=m}^{n-1}S_i\right) L_m
\ge \left(\prod_{i=m}^{n-1} S_i\right)\left((n-m)\log \gamma + L_m\right).
\end{equation*}
For any $\delta\in(0,1]$, one obtains
\begin{align*}
&\limsup_{n\to\infty}P\left(\log\log(1/X_n) > \mathbb{E}[\log S_1]n+t\sqrt{\mathbb{V}[\log S_1]n}+f(n)\right)\\
&\le\limsup_{n\to\infty}P\left(\log\log(1/X_n) > \mathbb{E}[\log S_1]n+t\sqrt{\mathbb{V}[\log S_1]n}+f(n),~X_m \le \delta\right)\\
&\quad+\limsup_{n\to\infty}P\left(\log\log(1/X_n) > \mathbb{E}[\log S_1]n+t\sqrt{\mathbb{V}[\log S_1]n}+f(n),~X_m > \delta\right)\\
&\le\limsup_{n\to\infty}P\left(\log\log(1/X_n) > \mathbb{E}[\log S_1]n+t\sqrt{\mathbb{V}[\log S_1]n}+f(n),~X_m \le \delta\right)\\
&\quad+\limsup_{n\to\infty}P\left(X_n < \frac{\delta}{2},~X_m > \delta\right)\\
&\le \limsup_{n\to\infty}
P\left(\sum_{i=m}^{n-1} \log S_i + \log\left(-(n-m)\log \gamma - L_m\right)
> \mathbb{E}[\log S_1]n+t\sqrt{\mathbb{V}[\log S_1]n}+f(n),~X_m\le\delta\right)\\
&\quad+P\left(X_\infty\le\frac{\delta}{2},~X_m>\delta\right)\\
&= 
Q(t)P(X_m\le\delta)+P\left(X_\infty\le\frac{\delta}{2},~X_m>\delta\right).
\end{align*}
The last equality follows from (c4) and the central limit theorem.
One obtains
\begin{align*}
&\limsup_{n\to\infty}P\left(\log\log(1/X_n) > \mathbb{E}[\log S_1]n+t\sqrt{\mathbb{V}[\log S_1]n}+f(n)\right)\\
&\le\limsup_{m\to\infty}\left\{Q(t)P(X_m\le\delta)+P\left(X_\infty\le\frac{\delta}{2},~X_m>\delta\right)\right\}\\
&\le Q(t)P(X_\infty\le\delta)+P\left(X_\infty\le\frac{\delta}{2},~X_\infty\ge\delta\right)
= Q(t)P(X_\infty\le\delta).
\end{align*}
By letting $\delta$ to 0, one obtains the result.
\end{proof}

\chapter{Polar Codes and its Construction}\label{chap:pcodes}
\section{Introduction}
Polar codes are channel codes based on the channel polarization phenomenon.
Polar codes achieve symmetric capacity under efficient encoding and decoding algorithms.
However, construction of polar codes requires high computational cost in the original work~\cite{5075875}.
One of the contribution of the thesis is to show for symmetric B-DMCs, a construction method 
with complexity $O(N)$ where $N$ is the blocklength~\cite{5205857}.

\section{Preliminaries}
For $x\in\{0,1\}$, $\bar{x}$ represents the bit flipping of $x$.
\begin{definition}[Symmetric B-DMC]
A B-DMC $W: \mathcal{X}\to\mathcal{Y}$ is said to be symmetric if there exists a permutation $\pi$ on $\mathcal{Y}$ such that
$W(\pi(y) \mid x) = W(y \mid \bar{x})$ for all $y\in\mathcal{Y}$.
\end{definition}
\begin{definition}
The error probability of a B-DMC $W$ is defined as
\begin{equation*}
P_e(W) := \frac12\sum_{y: W(y\mid 1)>W(y\mid 0)}\hspace{-1.5em} W(y\mid 0) +
\frac12\sum_{y: W(y\mid 1)<W(y\mid 0)}\hspace{-1.5em} W(y\mid 1) +
\frac12 \sum_{y: W(y\mid 1)=W(y\mid 0)} \hspace{-1.5em}W(y\mid 0)
\end{equation*}
\end{definition}
In order to bound the error probability of polar codes, Bhattacharyya parameter is useful.
\begin{lemma}\label{lem:PeZ}\cite{korada2009thesis}
\begin{equation*}
\frac12\left(1-\sqrt{1-Z(W)^2}\right)\le P_e(W)\le  \frac12 Z(W).
\end{equation*}
\end{lemma}
\section{Polar Codes}\label{sec:polar}
Polar codes are based on channel polarization phenomenon.
Fix an $\ell\times\ell$ matrix $G$, $\mathcal{F}\subseteq \{0,\dotsc,\ell^n-1\}$ and $u_\mathcal{F}$.
Variables belonging to $u_\mathcal{F}$ and $u_{\mathcal{F}^c}$ are called
\textit{frozen variables} and \textit{information variables}, respectively.
Let $G_n := (I_{\ell^{n-1}}\otimes G)R_{\ell,n}(I_\ell\otimes G_{n-1})$ where $\otimes$ denotes the Kronecker product,
where $R_{\ell,n}$ is a permutation matrix such that
$(u_0,\dotsc,u_{\ell^n-1})R_{\ell,n} =
(u_0, u_\ell,\dotsc,u_{\ell^{n-1}}, u_1, u_{\ell+1},\dotsc,u_{\ell^{n-1}+1},\dotsc,u_{\ell-1}, u_{2\ell-1},\dotsc, u_{\ell^n-1})$,
where $I_k$ denotes the identity matrix of size $k$,
and where $G_1 = G$.
An encoding result of a polar code of length $\ell^n$ is represented as $u_0^{\ell^n-1} G_n$
where $u_{\mathcal{F}^c}$ is constituted by pre-encoding values corresponding to a message.
Note that $G_n = B_{\ell,n}G^{\otimes n}$ where $B_{\ell,n}$ is the bit-reversal permutation matrix
with respect to $\ell$-ary expansion~\cite{5075875}.
More precisely, for $x_0^{\ell^n-1}=u_0^{\ell^n-1}B_{\ell,n}$,  
$x_i$ is equal to $u_j$ where $\ell$-ary expansion $b_1\dotsm b_n$ of $i$ is the reverse of $\ell$-ary expansion
$b_n\dotsm b_1$ of $j$.

We assume successive cancellation (SC) decoder for polar codes.
For $i\in\{0,\dotsc,\ell^n-1\}$, let
\begin{equation*}
W_n^{\langle i\rangle}(y_0^{\ell^n-1}, u_0^{i-1}\mid u_i) :=
\frac1{2^{\ell^n-1}}\sum_{u_{i+1}^{\ell^n-1}}
W^{\ell^n}(y_0^{\ell^n-1}\mid (\hat{u}_0^{i-1},\, u_i,\, u_{i+1}^{\ell^n-1}) G_n).
\end{equation*}
In SC decoding, all variables, which consist of information variables and frozen variables, are decoded sequentially from $u_0$ to $u_{\ell^n-1}$.
The decoding result for $u_i$ of SC decoder is
\begin{equation*}
\hat{U}_i(y_0^{\ell^n-1},\hat{u}_0^{i-1}) = \begin{cases}
u_i, &\text{if } i\in \mathcal{F}\\
\mathop{\rm argmax}_{u_i \in \{0,1\}} 
W_n^{\langle i\rangle}(y_0^{\ell^n-1},\hat{u}_0^{i-1}\mid u_i),& \text{if } i\notin\mathcal{F}
\end{cases}
\end{equation*}
where $\hat{u}_0^{i-1}$ is a result of SC decoding for $u_0^{i-1}$.
When 
$W_n^{\langle i\rangle}(y_0^{\ell^n-1},\hat{u}_0^{i-1}\mid 0)= W_n^{\langle i\rangle}(y_0^{\ell^n-1},\hat{u}_0^{i-1}\mid 1)$
for $i\in \mathcal{F}$, the decoding result is determined as 0 and 1 with probability one half.

\section{Error Probabilities of Polar Codes}


We now consider an expected error probability of polar codes where values of $u_\mathcal{F}$ are uniformly chosen from $\{0,1\}^{|\mathcal{F}|}$.
Let $(\Omega = \{0,1\}^{\ell^n}\times\mathcal{Y}^{\ell^n}, 2^\Omega, P)$ be a probability space where $P$ is
\begin{equation*}
P((u_0^{\ell^n-1}, y_0^{\ell^n-1})) := \frac1{2^{\ell^n}} W^{\ell^n}(y_0^{\ell^n-1}\mid u_0^{\ell^n-1} G_n).
\end{equation*}
%
Let $\mathcal{B}_i$ and $\mathcal{A}_i$ be
\begin{align*}
\mathcal{B}_i &:= \{(u_0^{\ell^n-1}, y_0^{\ell^n-1})\in\Omega \mid
 \hat{u}_0^{i-1} = u_0^{i-1}, \hat{U}_i(\hat{u}_0^{i-1}, y_0^{\ell^n-1}) \ne u_i\}\\
\mathcal{A}_i &:= \{(u_0^{\ell^n-1}, y_0^{\ell^n-1})\in\Omega \mid \hat{U}_i(u_0^{i-1}, y_0^{\ell^n-1}) \ne u_i\}.
\end{align*}
From the definition, one obviously sees $\mathcal{B}_i \subseteq \mathcal{A}_i$.
An expected error probability of polar codes where values of $u_\mathcal{F}$ are uniformly chosen from $\{0,1\}^{|\mathcal{F}|}$
is $P(\bigcup_{i\in \mathcal{F}^c} \mathcal{B}_i)$.
One obtains an upper bound of the expected error probability as
\begin{equation}
P\left(\bigcup_{i\in \mathcal{F}^c} \mathcal{B}_i\right)
=\sum_{i\in \mathcal{F}^c} P\left(\mathcal{B}_i\right)
\le\sum_{i\in \mathcal{F}^c} P\left(\mathcal{A}_i\right)
=\sum_{i\in \mathcal{F}^c} P_e(W^{\langle i\rangle}_n)
=\sum_{i\in \mathcal{F}^c} P_e(W^{(b_1)\dotsm (b_n)})
\le\sum_{i\in \mathcal{F}^c} Z(W^{(b_1)\dotsm (b_n)})
\label{eq:pub}
\end{equation}
where $\ell$-ary expansion of $i$ is $(b_1\dotsm b_n)$.
The last equality is not proven here.
If one chooses $\mathcal{F}^c = \{ i \in \{0,\dotsc,\ell^n-1\} \mid Z(W^{(i)}) < 2^{-\ell^{\beta n}}\}$,
the expected error probability is smaller than $\ell^n2^{-\ell^{\beta n}}$.
From Theorem~\ref{thm:rop}, $|\mathcal{F}^c|/\ell^n$ is close to $I(W)$ as $n\to\infty$ for any $\beta\in(0,E(G))$.
Hence, the expected error probability is $o(2^{-\ell^{\beta n}})$ for any $\beta\in(0,E(G))$
while coding rate is fixed and smaller than $I(W)$.
On the other hand, one obtains
\begin{equation*}
P\left(\bigcup_{i\in \mathcal{F}^c} \mathcal{B}_i\right) \ge \max_{i\in\mathcal{F}^c}P(\mathcal{A}_i)
= \max_{i\in\mathcal{F}^c}P_e(W^{(b_1)\dotsm(b_n)})
\ge \max_{i\in\mathcal{F}^c}\frac12\left(1- \sqrt{1- Z(W^{(b_1)\dotsm(b_n)})^2}\right).
\end{equation*}
Hence, the expected error probability is $\omega(2^{-\ell^{\beta n}})$ for any $\beta>E(G)$.
From Propositions~\ref{prop:rdirect} and \ref{prop:rconv}, one obtains the following result~\cite{tanaka2010rre}.
\begin{theorem}
There exists a sequence of polar codes such that coding rate tends to $R<I(W)$ and
the error probability is
\begin{equation*}
o\left(2^{-2^{E(G)n + Q^{-1}(R/I(W))\sqrt{V(G)n}+f(n)}}\right)
\end{equation*}
for any $f(n)=o(\sqrt{n})$.
The error probability of any sequence of polar codes where coding rate tends to $R<I(W)$ is
\begin{equation*}
\omega\left(2^{-2^{E(G)n + Q^{-1}(R/I(W))\sqrt{V(G)n}+\epsilon\sqrt{n}}}\right)
\end{equation*}
for any $\epsilon>0$.
\end{theorem}

We now consider asymptotic expected error probability of polar codes in a restricted class under maximum likelihood (ML) decoding.
Assume that the weight of $i$-th row of $G$ is equal to $D^{[i]}$.
Then, the weight of $i$-th row of $G^{\otimes n}$ is $\prod_{j=1}^n D_{b_j}$
where $(b_1\dotsm b_n)$ is an $\ell$-ary expansion of $i$.
Fraction of rows which satisfy $\sum_{j=1}^n\log_\ell D_{b_j} > nE(G)+t\sqrt{V(G)n}$
tends to $Q(t)$ from the central limit theorem.
Since the error probability of ML decoding is lower bounded by $P_e(W)^{D}$ where $D$ is the minimum distance of the code,
expected error probability of polar codes on ML decoding is
$\omega(2^{-\ell^{E(G) n+Q^{-1}(R)\sqrt{V(G)n}+\epsilon\sqrt{n}}})$ for any $\epsilon>0$.

\section{Complexities}\label{sec:complexity}
\subsection{Complexity of encoding}
Since encoding procedure of polar codes is multiplication of a matrix, the complexity of encoding is $O(\ell^{2n})$.
Further, since the matrix $G^{\otimes n}$ has recursive structure, the complexity is reduced like the fast Fourier transform.
Let $c$ denote the complexity of evaluation of $w_0^{\ell-1}G$.
Let $d$ denote the complexity of evaluation of $w_0^{\ell^n-1}R_{\ell,n}$ divided by $\ell^n$.
Let $\chi_E(n)$ denote the complexity of evaluation of $u_0^{\ell^n-1}G_n$.
Since $G_n = (I_{\ell^{n-1}}\otimes G)R_{\ell,n}(I_\ell\otimes G_{n-1})$,
one obtains
$\chi_E(n) = \ell^{n-1}c +\ell^nd + \ell \chi_E(n-1)$.
Hence, $\chi_E(n) = O(n\ell^n)$.

\subsection{Complexity of decoding}
SC decoding can be described as
\begin{equation*}
\hat{U}_i(\hat{u}_0^{i-1}, y_0^{\ell^n-1}) = \begin{cases}
u_i, &\text{if } i\in \mathcal{F}\\
0, & \text{if } i\notin\mathcal{F},~L_n^{\langle i\rangle}(y_0^{\ell^n-1}, \hat{u}_0^{i-1}) > 0\\
1, & \text{if } i\notin\mathcal{F},~L_n^{\langle i\rangle}(y_0^{\ell^n-1}, \hat{u}_0^{i-1}) < 0
\end{cases}
\end{equation*}
where
\begin{equation*}
L_n^{\langle i\rangle}(y_0^{\ell^n-1}, \hat{u}_0^{i-1})
:=\log\frac{W_n^{\langle i\rangle}(y_0^{\ell^n-1},\hat{u}_0^{i-1}\mid 0)}{W_n^{\langle i\rangle}(y_0^{\ell^n-1},\hat{u}_0^{i-1}\mid 1)}
\end{equation*}
is the log likelihood ratio (LLR) of $u_i$.
Let
$\mathcal{G}(l_0^{\ell-1}) := r_0^{\ell-1}$ for $l_0^{\ell-1}\in\mathbb{R}^\ell$
where $r_i$ denotes the LLR of $u_i$ given $u_0^{i-1}$ when an LLR of $u_0^{\ell-1}G$ is $l_0^{\ell-1}$.
Let $\chi_D(n)$ denote the number of evaluation of $\mathcal{G}$ in calculation of $\{L_n^{\langle i\rangle}\}_{i\in\{0,\dotsc,\ell^n-1\}}$.
Since $U_{\ell m+i}\to G(U_{\ell m}^{\ell m+\ell -1})\to Y_0^{\ell^n-1}$ for all $i\in\{0,\dotsc,\ell-1\}$,
it holds that
$\chi_D(n) = \ell^{n-1} + \ell \chi_D(n-1)$.
Hence, $\chi_D(n) = O(n\ell^n)$.

\subsection{Complexity of construction}

Construction of a polar code is equivalent to selection of a set $\mathcal{F}$ of frozen variables.
In~\cite{5075875}, Ar{\i}kan proposed a criterion on which $i$ with small $Z(W^{\langle i\rangle}_n)$ are chosen as information variables
in order to minimize the upper bound~\eqref{eq:pub}.
However, unless $W$ is the binary erasure channel (BEC), the complexity of the evaluation of
$Z(W^{\langle i\rangle}_n)$ is exponential in the blocklength.
In order to avoid the high cost of computation, he also proposed a Monte-Carlo method
which estimates $Z(W^{\langle i\rangle}_n)$ by numerical simulations.
Ar{\i}kan also proposed a heuristic method in which a B-DMC $W$ is regarded as the BEC of erasure probability $1-I(W)$~\cite{4542778}.
However, polar codes constructed by these methods do not provably achieve symmetric capacity.
In this chapter, we describe a novel construction method for any symmetric B-DMC
whose complexity is linear in the blocklength~\cite{5205857},~\cite{5166430}.
Polar codes constructed by the method provably achieve symmetric capacity.
The method is based on \textit{density evolution}, which has been used for evaluation of
the large blocklength limit of the bit error probability of LDPC codes.

\begin{figure}[t]
\psfrag{0}{\hspace{-1em}\tt 000}
\psfrag{1}{\hspace{-1em}\tt 001}
\psfrag{2}{\hspace{-1em}\tt 010}
\psfrag{3}{\hspace{-1em}\tt 011}
\psfrag{4}{\hspace{-1em}\tt 100}
\psfrag{5}{\hspace{-1em}\tt 101}
\psfrag{6}{\hspace{-1em}\tt 110}
\psfrag{7}{\hspace{-1em}\tt 111}
\includegraphics[width=0.45\hsize]{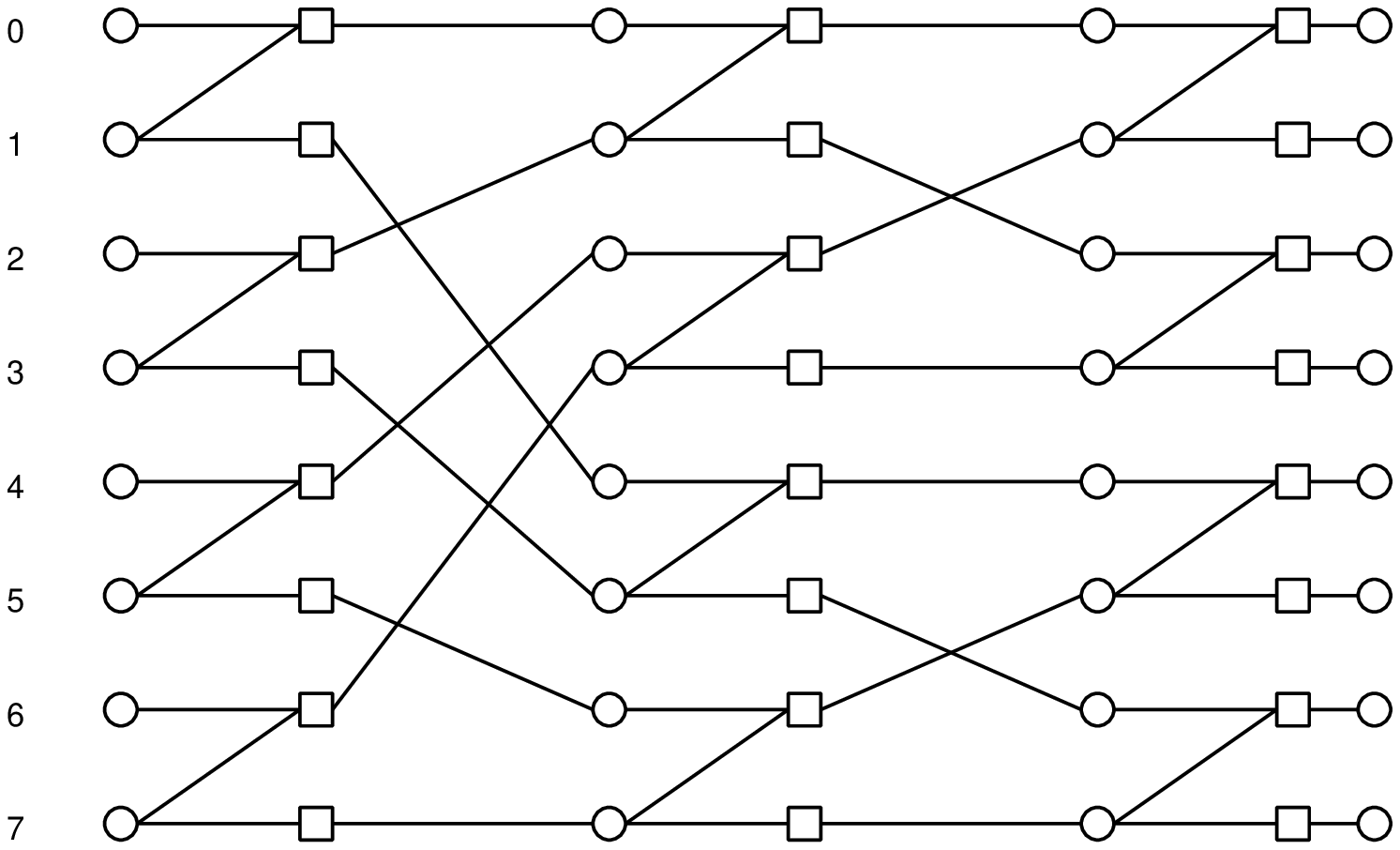}\hspace{2em}
\includegraphics[width=0.45\hsize]{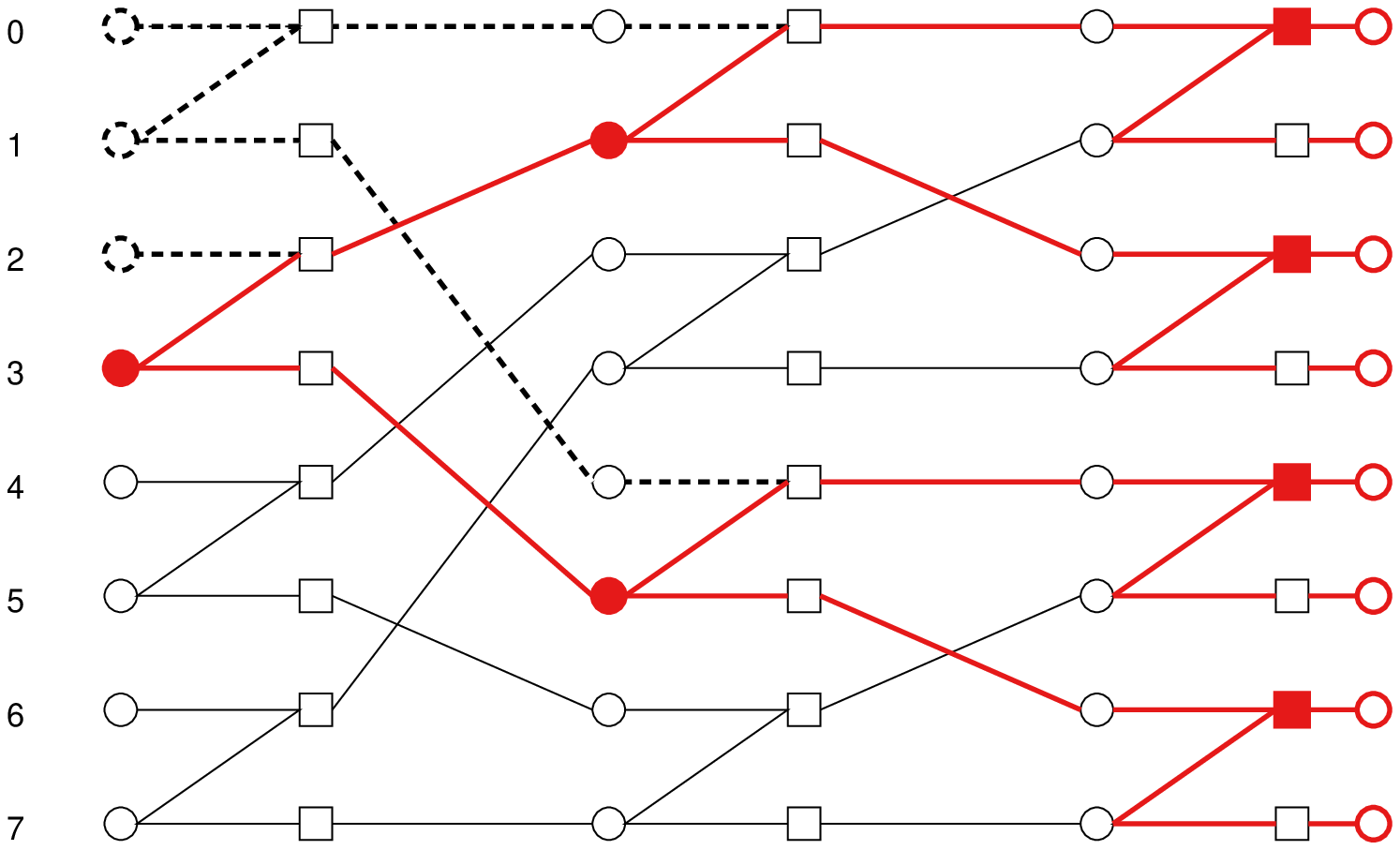}
\caption{Left: Factor graph of $G_3$. Right: Decoding graph of $u_3$.}
\label{fig:factor}
\end{figure}

\section{Factor Graphs, Belief Propagation and Density Evolution}
Factor graphs, belief propagation (BP), and density evolution are important tools used in certain areas.
The book of Richardson and Urbanke is a good reference~\cite{RiU05/LTHC}. 
A factor graph is a graph which represents a probability distribution.
The left panel of Figure~\ref{fig:factor} shows the factor graph of $BG^{\otimes 3}$ when $G$ is the $2\times 2$ matrix~\eqref{eq:2x2}.
Belief propagation is an efficient algorithm for calculation of marginal probability distributions on a tree factor graph.
SC decoding can be regarded as BP decoding on a tree graph as in the right panel of Figure~\ref{fig:factor}.

Density evolution is a method which recursively evaluates probability distributions of messages on a tree graph.
Let $W$ be a symmetric B-DMC\@.
There exists a probability density function $\mathsf{a}_W$ on $(-\infty,+\infty]$ of an LLR when 0 is transmitted,
which is linear combination of the Dirac delta function.
When $W$ is the BEC of erasure probability $\epsilon$, $\mathsf{a}_W = (1-\epsilon)\delta_\infty + \epsilon \delta_0$
where $\delta_x$ is the Dirac delta function centered at $x$.
When probability density functions of input messages of variable nodes (respectively check nodes) are $\mathsf{a}$ and $\mathsf{b}$,
the probability density function of the output message is denoted by
$\mathsf{a} \varoast \mathsf{b}$ (respectively $\mathsf{a} \boxast \mathsf{b}$).
Details of density evolution is written in~\cite{RiU05/LTHC}.

\section{Construction using Density Evolution}
In this section, for simplicity, we assume that $G$ is the $2\times 2$ matrix~\eqref{eq:2x2}. 
We consider using density evolution for evaluation of $P_e(W_n^{\langle i\rangle})$ for $i\in\{0,\dotsc,\ell^n-1\}$.
In fact, we can evaluate the probability density function of an LLR of $W^{\langle i\rangle}_n$ by density evolution~\cite{5205857}.
\begin{theorem}
For $n\ge 1$,
\begin{align*}
\mathsf{a}_{W_{n}^{\langle i\rangle}} &= \mathsf{a}_{W_{n-1}^{\langle (i-1)/2\rangle}}\varoast\mathsf{a}_{W_{n-1}^{\langle (i-1)/2\rangle}},&
\text{if $i$ is odd} \\
\mathsf{a}_{W_{n}^{\langle i\rangle}} &= \mathsf{a}_{W_{n-1}^{\langle i/2\rangle}}\boxast\mathsf{a}_{W_{n-1}^{\langle i/2\rangle}},&
\text{if $i$ is even}.
\end{align*}
\end{theorem}
\noindent
$P_e(W_n^{\langle i\rangle})$ is obtained by an appropriate integration of $\mathsf{a}_{W_n^{\langle i\rangle}}$.

Let us consider the number $\chi_C(n)$ of operations $\varoast$ and $\boxast$ in
the calculation of $\{\mathsf{a}_{W_n^{\langle i\rangle}}\}_{i=0,\dotsc,2^n-1}$.
In order to calculate $\{\mathsf{a}_{W_n^{\langle i\rangle}}\}_{i=0,\dotsc,2^n-1}$, 
calculation of $\{\mathsf{a}_{W_{n-1}^{\langle i\rangle}}\}_{i=0,\dotsc,2^{n-1}-1}$ is required.
Further, $2^n$ operations of $\varoast$ and $\boxast$ are necessary.
Hence,
\begin{equation*}
\chi_C(n) = 2^n + \chi_C(n-1).
\end{equation*}
This implies $\chi_C(n) = O(2^n)$ meaning that it is proportional to the blocklength.
It is known that the complexity of selection of the $s$ smallest elements from a set of size $t$ is $O(t)$.
Hence, 
the complexity of construction is linear in the blocklength
if we assume that the complexity of the operations $\varoast$ and $\boxast$ is constant.
However, the required precision increases as the blocklength increases.
When $W$ is the binary symmetric channel (BSC), the number of mass points grows exponentially in the blocklength.
It has not been well known how quantization errors affect performance of resulting codes.

\section{Numerical Calculation and Simulation}
In this section, error probability of polar codes constructed by using density evolution and 
error probability of polar codes constructed by Ar{\i}kan's heuristic method~\cite{4542778}, in which
$W$ is regarded as BEC of the same capacity are compared.
Figure~\ref{fig:comp} shows results for the BSC with crossover probability 0.11 and the blocklength is 4096.
The capacity of the BSC is 0.5.
The error probabilities of polar codes which are constructed by using density evolution are much smaller than
the error probabilities of polar codes which are constructed by the heuristic method.
This result implies that information variables should be chosen by taking into account the channel, rather than its capacity only.
This can easily be confirmed via the simplest case with $n=2$:
The error probability $P_e(W^{\langle 1\rangle}_2)$ is less than, equal to, and larger than $P_e(W^{\langle 2\rangle}_2)$
when the channel is the BEC, BSC,
and binary additive white Gaussian noise channel (BAWGNC), respectively,
irrespective of the channel parameters.
In~\cite{hassami2009ccp}, \cite{korada2009thesis}, the authors show that
polar codes and SC decoding do not achieve symmetric capacity universally.

\begin{figure}[htb]
\psfrag{x}{\hspace{-3em}Coding rate}
\psfrag{y}{\hspace{-5em}The error probability}
\includegraphics[width=0.8\hsize]{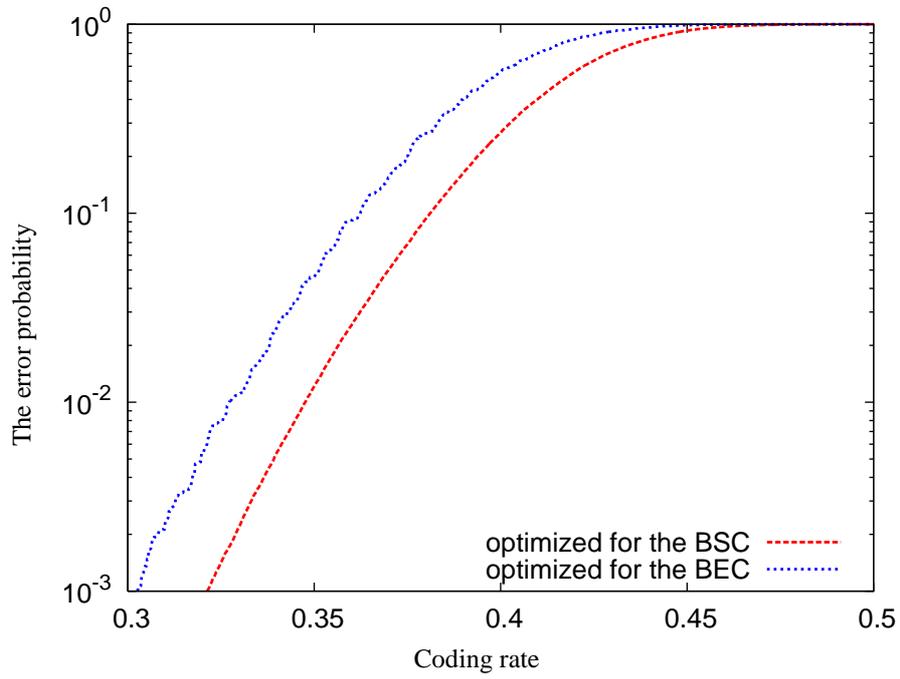}
\caption{Comparison of the error probability of polar codes constructed by different methods.
The bottom curve is the result of construction using density evolution.
The top curve is the result of construction using the heuristic method of Ar{\i}kan~\cite{4542778}.
The channel is the BSC of crossover probability 0.11.
The capacity is 0.5.
The blocklength is 4096.
}
\label{fig:comp}
\end{figure}

\chapter{Channel Polarization of $q$-ary DMC by Arbitrary Kernel}\label{chap:qary}
\section{Introduction}
{\c S}a{\c s}o{\u g}lu, Telatar and Ar{\i}kan considered channel polarization of $q$-ary channels~\cite{sasoglu2009pad}.
They regarded $\mathcal{X}$ as $\mathbb{Z}/q\mathbb{Z}$ and assumed that the size $\ell$ of channel transform is 2.
They showed that the channel polarization phenomenon occurs on the $2\times 2$ matrix~\eqref{eq:2x2} when $q$ is prime,
and that using randomized permutations, the channel polarization phenomenon occurs for any $q$.
In this chapter, we consider channel polarization for arbitrary $q$ and arbitrary channel transform~\cite{mori2010cpq}.

\section{Preliminaries}
In this chapter, we assume that $|\mathcal{X}|=q$ and that the base of logarithm is $q$ unless otherwise stated.
Let $\mathrm{e}$ denote the base of the natural logarithm.
\begin{definition}
The symmetric capacity of a $q$-ary input channel $W:\mathcal{X}\to\mathcal{Y}$ is defined as
\begin{equation*}
I(W) := \sum_{x\in\mathcal{X}}\sum_{y\in\mathcal{Y}} \frac1q W(y\mid x)
\log\frac{W(y\mid x)}{\frac1q \sum_{x'\in\mathcal{X}}W(y\mid x')}.
\end{equation*}
Note that $I(W)\in[0,1]$.
\end{definition}
\begin{definition}
Let $\mathcal{D}_x:=\{y\in\mathcal{Y}\mid W(y\mid x) > W(y\mid x'), \forall x'\in\mathcal{X}, x'\ne x\}$.
The error probability of $W$ is defined as
\begin{equation*}
P_e(W) := \frac1q \sum_{x\in\mathcal{X}} \sum_{y\in\mathcal{D}_x^c} W(y\mid x).
\end{equation*}
\end{definition}
\begin{definition}
The Bhattacharyya parameter of $W$ is defined as
\begin{equation*}
Z(W) := \frac1{q(q-1)} \sum_{\substack{x\in\mathcal{X}, x'\in\mathcal{X},\\ x\ne x'}} Z_{x,x'}(W)
\end{equation*}
where Bhattacharyya parameter between $x$ and $x'$ is defined as
\begin{equation*}
Z_{x,x'}(W) := \sum_{y\in\mathcal{Y}} \sqrt{W(y\mid x)W(y\mid x')}.
\end{equation*}
Note that $Z(W)\in[0,1]$ and that $Z_{x,x'}(W)\in[0,1]$.
\end{definition}
\begin{lemma}\label{lem:tri}
For any $x\in\mathcal{X}$, $x'\in\mathcal{X}$ and $x''\in\mathcal{X}$,
\begin{equation*}
\sqrt{1-Z_{x,x'}}\le \sqrt{1-Z_{x,x''}} + \sqrt{1-Z_{x'',x'}}.
\end{equation*}
\end{lemma}
\begin{proof}
The inequality follows from the triangle inequality of Euclidean distance since
\begin{equation*}
\sqrt{1-Z_{x,x'}} = \sqrt{\frac12\sum_{y\in\mathcal{Y}}\left(\sqrt{W(y\mid x)} - \sqrt{W(y\mid x')}\right)^2}.
\end{equation*}
\end{proof}
\begin{lemma}
\begin{equation*}
P_e(W)\le (q-1)Z(W)
\end{equation*}
\end{lemma}
\begin{lemma}\label{lem:IZ}\cite{sasoglu2009pad}
\begin{align*}
I(W)&\ge\log\frac{q}{1+(q-1)Z(W)}\\
I(W)&\le\log(q/2)+(\log2)\sqrt{1-Z(W)^2}\\
I(W)&\le 2(q-1)(\log\,\mathrm{e})\sqrt{1-Z(W)^2}.
\end{align*}
\end{lemma}
\begin{definition}
The maximum and the minimum of the Bhattacharyya parameters between two alphabets are defined as
\begin{align*}
Z_\text{\rm max}(W) &:= \max_{x\in\mathcal{X},x'\in\mathcal{X},x\ne x'} Z_{x,x'}(W)\\
Z_\text{\rm min}(W) &:= \min_{x\in\mathcal{X},x'\in\mathcal{X}} Z_{x,x'}(W).
\end{align*}
Let $\sigma:\mathcal{X}\to\mathcal{X}$ be a permutation.
Let $\sigma^i$ denote the $i$-th power of $\sigma$.
The average Bhattacharyya parameter of $W$ between $x$ and $x'$ with respect to $\sigma$
is defined as the average of $Z_{z,z'}(W)$ over
the subset $\{(z,z')=(\sigma^i(x),\sigma^i(x'))\in\mathcal{X}^2\mid i=0,1,\ldots,q!-1\}$
\begin{align*}
Z_{x,x'}^{\sigma}(W) &:= \frac1{q!} \sum_{i=0}^{q!-1} Z_{\sigma^i(x),\sigma^i(x')}(W).
\end{align*}
\end{definition}

\section{Channel Polarization on $q$-ary Channels}
We consider channel transform using a one-to-one onto mapping $g: \mathcal{X}^\ell\to\mathcal{X}^\ell$,
called a kernel.
\begin{definition}
Let D-MC $W:\mathcal{X}\to\mathcal{Y}$.
Then D-MC $W^\ell :\mathcal{X}^\ell\to\mathcal{Y}^\ell$,
$W^{(i)}:\mathcal{X}\to\mathcal{Y}^\ell\times\mathcal{X}^{i-1}$, and $W_{u_0^{i-1}}^{(i)}: \mathcal{X}\to\mathcal{Y}^\ell$ are defined as
\begin{align*}
W^\ell(y_0^{\ell-1}\mid x_0^{\ell-1}) &:=\prod_{i=0}^{\ell-1} W(y_i\mid x_i)\\
W^{(i)}(y_0^{{\ell-1}},u_0^{i-1}\mid u_i) &:= \frac1{q^{\ell-1}}
\sum_{u_{i+1}^{\ell-1}} W^{\ell}(y_0^{{\ell-1}}\mid g(u_0^{{\ell-1}}))\\
W_{u_0^{i-1}}^{(i)}(y_0^{{\ell-1}}\mid u_i) &:= \frac1{q^{{\ell}-i-1}}
\sum_{u_{i+1}^{\ell-1}} W^{{\ell}}(y_0^{{\ell-1}}\mid g(u_0^{{\ell-1}})).
\end{align*}
\end{definition}

Assume that $\{B_i\}_{i\in\mathbb{N}}$ is a sequence of independent uniform random variables taking values on $\{0,\dotsc,\ell-1\}$.
In the probabilistic channel transform $W\to W^{(B_i)}$,
expectation of the symmetric capacity is invariant due to the chain rule for mutual information.
The following lemma is a consequence of the martingale convergence theorem~\cite{billingsley1995probability}.
\begin{lemma}\label{lem:mc}
There exists a random variable $I_\infty$ such that $I(W^{(B_1)\dotsm(B_n)})$ converges to $I_\infty$
almost surely as $n\to\infty$.
\end{lemma}


From Lemma~\ref{lem:IZ}, $I(W)$ is close to 0 and 1 when $Z(W)$ is close to 1 and 0, respectively.
In order to show channel polarization, i.e., $I_\infty \in \{0,1\}$ with probability 1,
it suffices to show $\lim_{n\to\infty} P(Z(W^{(B_1)\dotsm(B_n)})\in(\delta,1-\delta)) = 0$ for any $\delta\in(0,1/2)$.
The following lemma is useful for this purpose.
\begin{lemma}\label{lem:polar}
Let $\{\mathcal{Y}_n\}_{n\in\mathbb{N}}$ be a random process taking values on a discrete set.
Let $\{W_n:\mathcal{X}\to\mathcal{Y}_n\}_{n\in\mathbb{N}}$ be a random process taking values on $q$-ary DMC\@.
Let $\sigma$ and $\tau$ be permutations on $\mathcal{X}$.
Let
\begin{equation*}
W_n'(y_1, y_2 \mid x) = \frac1q W_n(y_1\mid \sigma(x))W_n(y_2\mid \tau(x)).
\end{equation*}
Assume
\begin{equation*}
\lim_{n\to\infty} |I(W_n')-I(W_n)|=0
\end{equation*}
with probability 1.
Then
$\lim_{n\to\infty}P(Z^{\tau\sigma^{-1}}_{x,x'}(W_n) \in (\delta,1-\delta)) =0$
for any $x\in\mathcal{X}$, $x'\in\mathcal{X}$ and $\delta\in(0,1/2)$.
\end{lemma}
\begin{proof}
Let $Z$, $Y_1$ and $Y_2$ be random variables which take values on $\mathcal{X}$, $\mathcal{Y}_n$ and $\mathcal{Y}_n$, respectively,
and jointly obey the distribution
\begin{equation*}
P_n(Z=z,\,Y_1=y_1,Y_2=y_2)
=\frac1q W_n(y_1\mid \sigma(z))W_n(y_2\mid \tau(z)).
\end{equation*}
Since $I(W_n') = I(Z;Y_1,Y_2)$ and $I(W_n) = I(Z;Y_1)$,
$I(Z;Y_1,Y_2)-I(Z;Y_1)
=I(Z;Y_2\mid Y_1)$
tends to 0 with probability 1 by the assumption.
Since the mutual information is lower bounded by the cutoff rate as shown in Proposition~\ref{prop:cutoff},
one obtains
\begin{align*}
I(Z;Y_2\mid Y_1)
&\ge
-\log \sum_{y_1\in\mathcal{Y}_n,y_2\in\mathcal{Y}_n} P_n(Y_1=y_1)
\Bigg[\sum_{z\in\mathcal{X}} P_n(Z=z\mid Y_1=y_1)
\sqrt{P_n(Y_2=y_2\mid Z=z, Y_1=y_1)}\Bigg]^2\nonumber\\
&=
-\log \sum_{y_1\in\mathcal{Y}_n,z\in\mathcal{X},x\in\mathcal{X}} P_n(Y_1=y_1)P_n(Z=z\mid Y_1=y_1)
P_n(Z=x\mid Y_1=y_1)
 Z_{\tau(z),\tau(x)}(W_n)\nonumber\\
&=
-\log \sum_{y_1\in\mathcal{Y}_n,z\in\mathcal{X},x\in\mathcal{X}} q_n(y_1,z,x)
 Z_{\tau(\sigma^{-1}(z)),\tau(\sigma^{-1}(x))}(W_n)
\end{align*}
where
\begin{equation*}
q_n(y_1,z,x)
 := P_n(Y_1=y_1)
 P_n(Z=\sigma^{-1}(z)\mid Y_1=y_1)
 P_n(Z=\sigma^{-1}(x)\mid Y_1=y_1).
\end{equation*}
Since
\begin{align*}
\sum_{y_1\in\mathcal{Y}} q_n(y_1,z,x) &= \sum_{y_1\in\mathcal{Y}}P_n(Y_1=y_1)
\left(\sqrt{P_n(Z=\sigma^{-1}(z)\mid Y_1=y_1)
 P_n(Z=\sigma^{-1}(x)\mid Y_1=y_1)}\right)^2\\
&\ge
\bigg(\sum_{y_1\in\mathcal{Y}}P_n(Y_1=y_1)
\sqrt{P_n(Z=\sigma^{-1}(z)\mid Y_1=y_1)
 P_n(Z=\sigma^{-1}(x)\mid Y_1=y_1)}\bigg)^2\\
&=\frac1{q^2} Z_{z,x}(W_n)^2
\end{align*}
it holds
\begin{equation*}
I(Z;Y_2\mid Y_1)
\ge -\log \left[1-\frac1{q^2}\sum_{\substack{z\in\mathcal{X}, x\in\mathcal{X}\\z\ne x}}  
Z_{z,x}(W_n)^2 \left(1-Z_{\tau(\sigma^{-1}(z)),\tau(\sigma^{-1}(x))}(W_n)\right)\right].
\end{equation*}
The convergence of $I(Z;Y_2\mid Y_1)$ to 0 with probability 1 implies that
\begin{equation*}
Z_{z,x}(W_n)^2 \left(1-Z_{\tau(\sigma^{-1}(z)),\tau(\sigma^{-1}(x))}(W_n)\right)
\end{equation*}
tends to 0 with probability 1 for any $(z,x)\in\mathcal{X}^2$.
It consequently implies $\lim_{n\to\infty} P(Z_{z,x}^{\tau\sigma^{-1}}(W_n) \in (\delta,1-\delta))=0$
for any $(z,x)\in\mathcal{X}^2$ and $\delta\in(0,1/2)$.
\end{proof}

\begin{corollary}\label{cor:nonl}
Assume that
there exists $u_0^{\ell-2}\in\mathcal{X}^{\ell-1}$,
$(i,j)\in \{0,1,\dotsc,\ell-1\}^2$
and permutations $\sigma$ and $\tau$ on $\mathcal{X}$ such that
$i$-th element of $g(u_0^{\ell-1})$ and $j$-th element of $g(u_0^{\ell-1})$ are $\sigma(u_{\ell-1})$ and $\tau(u_{\ell-1})$, respectively,
and such that for any $v_0^{\ell-2}\ne u_0^{\ell-2}\in\mathcal{X}^{\ell-1}$
there exists $m\in\{0,1,\dotsc,\ell-1\}$ and a permutation $\mu$ on $\mathcal{X}$ such that
$m$-th element of $g(v_0^{\ell-1})$ is $\mu(v_{\ell-1})$.
Then, $\lim_{n\to\infty} P(Z^{\tau\sigma^{-1}}_{x,x'}(W_n)\in(\delta,1-\delta))=0$ for all $x\in\mathcal{X}$, $x'\in\mathcal{X}$
and $\delta\in(0,1/2)$.
\end{corollary}
\begin{proof}
Since $I(W^{(B_1)\dotsm(B_n)})$ converges to $I_\infty$ with probability 1,
$|I(W^{(B_1)\dotsm(B_{n})(\ell-1)})-I(W^{(B_1)\dotsm(B_n)})|$ has to converge to 0 with probability 1.
Let $U_0^{\ell-1}$ and $Y_0^{\ell-1}$ denote random variables ranging over $\mathcal{X}^{\ell}$ and $\mathcal{Y}^{\ell}$,
and obeying the distribution
\begin{equation*}
P(U_0^{\ell-1}=u_0^{\ell-1},\, Y_0^{\ell-1}=y_0^{\ell-1}) = \frac1q W^{(\ell-1)}(y_0^{\ell-1},u_0^{\ell-2}\mid u_{\ell-1}).
\end{equation*}
Then, it holds
\begin{align*}
I(W^{(\ell-1)}) &= I(Y_0^{\ell-1}, U_0^{\ell-2}; U_{\ell-1})\\
&= I(Y_0^{\ell-1}; U_{\ell-1}\mid U_0^{\ell-2})\\
&=\sum_{u_0^{\ell-2}}\frac1{q^{\ell-1}} I(Y_0^{\ell-1}; U_{\ell-1}\mid U_0^{\ell-2}=u_0^{\ell-2}).
\end{align*}
From the assumption,
$I(Y_0^{\ell-1}; U_{\ell-1}\mid U_0^{\ell-2}=u_0^{\ell-2})\ge I(W)$ for all $u_0^{\ell-2}\in\mathcal{X}^{\ell-1}$.
Hence,
$|I(W^{(B_1)\dotsm(B_{n})'})-I(W^{(B_1)\dotsm(B_n)})|$ has to converge to 0 with probability 1.
By applying Lemma~\ref{lem:polar}, one obtains the result.
\end{proof}
When $q=2$,
Corollary~\ref{cor:nonl} is sufficient to show the channel polarization phenomenon.
The derivation does not use linearity of a kernel.
When we assume that $\mathcal{X}$ is a finite field and that a kernel $g$ is linear,
the matrix $G$ representing the kernel $g$ is assumed to be lower triangular due to the same reason as in Chapter~\ref{chap:bpolar}.

\begin{theorem}\label{thm:prime}
Assume that $\mathcal{X}$ is a prime field, and that a linear kernel $G$ is not diagonal.
Then, $P(I_\infty\in\{0,1\})=1$.
\end{theorem}
\begin{proof}
Let $k$ be the largest number such that the number of non-zero elements in $k$-th row of $G$ is larger than 1.
Without loss of generality, we assume $G_{kk}=1$.
It holds
\begin{equation*}
W^{(k)}(y_0^{\ell-1}, u_0^{k-1} \mid u_k) = \frac1{q^{\ell-1}}
 \prod_{j=k+1}^{\ell-1} \left(\sum_{x\in\mathcal{X}} W(y_j \mid x)\right)
 \prod_{j\in S_0} W(y_j\mid x_j) \prod_{j\in S_1} W(y_j\mid G_{kj} u_k + x_j)
\end{equation*}
where $S_0 := \{j \in\{0,\dotsc,\ell-1\}\mid G_{kj} = 0\}$,
$S_1 := \{j \in\{0,\dotsc,\ell-1\}\mid G_{kj} \ne 0\}$, and
$x_j$ is $j$-th element of $(u_0^{k-1}, 0_k^{\ell-1})G$ where $0_k^{\ell-1}$ is all-zero vector of length $\ell-k$.
Let $m\in\{0,\dotsc,k-1\}$ be an arbitrary index such that $G_{km}\ne 0$.
Since each $u_0^{k-1}$ occurs with positive probability $1/q^{k}$,
we can apply Lemma~\ref{lem:polar} with $\sigma(x) = x$ and $\tau(x) = G_{km}x + z$ for an arbitrary $z\in\mathcal{X}$.
Hence, for sufficiently large $n$, $Z_{x,x'}^\tau(W^{(B_1)\dotsm(B_n)})$ is close to 0 or 1 almost surely
where $\tau(x) = G_{km}^ix +z$ for all $i\in\{0,\dotsc,q-2\}$ and all $z\in\mathcal{X}$.
Since $q$ is prime,
for any $x\in\mathcal{X}$ and $x'\in\mathcal{X}$ where $x\ne x'$,
$Z_{x,x'}^\tau(W^{(B_1)\dotsm(B_n)})$ is close to 1 if and only if
$Z(W^{(B_1)\dotsm(B_n)})$ is close to 1,
where $\tau(z)=z+x'-x$.
\end{proof}
This result is a simple generalization of the special case considered by {\c S}a{\c s}o{\u g}lu, Telatar and Ar{\i}kan~\cite{sasoglu2009pad}.
We also show another sufficient condition for channel polarization in the following corollary.
\begin{corollary}\label{cor:ff}
Assume that $\mathcal{X}$ is a field and that a linear kernel $G$ is not diagonal.
Let $k$ be the largest number such that the number of non-zero elements in $k$-th row of $G$ is larger than 1.
If there exists $j\in\{0,\dotsc,k-1\}$ such that $G_{kj}/G_{kk}$ is a primitive element,
it holds $P(I_\infty\in\{0,1\})=1$.
\end{corollary}
\begin{proof}
By applying Lemma~\ref{lem:polar}, one sees that
$\lim_{n\to\infty} P(Z^{\sigma}_{x,x'}(W^{(B_1)\dotsm(B_n)})\in(\delta,1-\delta))=0$ for all $x\in\mathcal{X}$, $x'\in\mathcal{X}$
and $\delta\in(0,1/2)$,
where $\sigma(x) = (G_{kj}/G_{kk})x + z$ for an arbitrary $z\in\mathcal{X}$.
It suffices to show that
for any $x\in\mathcal{X}$ and $x'\in\mathcal{X}$, $x\ne x'$,
$Z_{x,x'}(W^{(B_1)\dotsm(B_n)})$ is close to 1 if and only if
$Z(W^{(B_1)\dotsm(B_n)})$ is close to 1.
When $Z_{x,x'}(W^{(B_1)\dotsm(B_n)})$ is close to 1,
$Z_{0,(G_{kj}/G_{kk})(x'-x)}(W^{(B_1)\dotsm(B_n)})$ is close to 1.
Hence,
$Z_{0,(G_{kj}/G_{kk})^i(x'-x)}(W^{(B_1)\dotsm(B_n)})$ is close to 1 for any $i\in\{0,\dotsc,q-2\}$.
Since $G_{kj}/G_{kk}$ is a primitive element,
$Z_{0,x}(W^{(B_1)\dotsm(B_n)})$ is close to 1 for any $x\in\mathcal{X}$.
From Lemma~\ref{lem:tri}, it completes the proof.
\end{proof}

\section{Speed of Polarization}
The result in Chapter~\ref{chap:speed} is also applicable to non-binary channel polarization.
\begin{definition}
Partial distance of a kernel $g: \mathcal{X}^\ell\to\mathcal{X}^\ell$ is defined as
\begin{equation*}
D_{x,x'}^{[i]}(u_0^{i-1})
:= \min_{v_{i+1}^{\ell-1}, w_{i+1}^{\ell-1}}
d(g(u_0^{i-1},x,v_{i+1}^{\ell-1}),\,g(u_0^{i-1},x',w_{i+1}^{\ell-1}))
\end{equation*}
where $d(a,b)$ denotes the Hamming distance between $a\in\mathcal{X}^\ell$ and $b\in\mathcal{X}^\ell$.
\end{definition}
We also use the following quantities.
\begin{align*}
D_{x,x'}^{[i]} &:= \min_{u_0^{i-1}} D_{x,x'}^{[i]}(u_0^{i-1}),&
D_\text{max}^{[i]} &:= \max_{x\in\mathcal{X},x'\in\mathcal{X}} D_{x,x'}^{[i]},&
D_\text{min}^{[i]} &:= \min_{\substack{x\in\mathcal{X},x'\in\mathcal{X}\\x\ne x'}} D_{x,x'}^{[i]}.
\end{align*}
When $g$ is linear, $D^{[i]}_{x,x'}(u_0^{i-1})$ does not depend on $x$, $x'$ or $u_0^{i-1}$,
in which case
we will use the notation $D^{[i]}$ instead of $D^{[i]}_{x,x'}(u_0^{i-1})$.
For a full-rank square matrix $G$, $E(G)$ and $V(G)$ are defined in the same way as in Definition~\ref{def:exp}.

In order to apply the method in Chapter~\ref{chap:speed}, the following lemma similar to Lemma~\ref{lem:bZD} is used.

\begin{lemma}\label{lem:ZD}
\begin{equation*}
\frac1{q^{2(\ell-1-i)}}Z_\text{min}(W)^{D_{x,x'}^{(i)}(u_0^{i-1})}
\le Z_{x,x'}(W^{(i)}_{u_0^{i-1}}) \le q^{\ell-1-i}Z_\text{max}(W)^{D_{x,x'}^{(i)}(u_0^{i-1})}
\end{equation*}
\end{lemma}
\begin{proof}
Proof of the second inequality is almost the same as the proof in~\cite{korada2009pcc}.
\begin{align*}
Z_{x,x'}(W_{ u_0^{i-1}}^{(i)}) &= \sum_{y_0^{\ell-1}} \sqrt{W_{u_0^{i-1}}^{(i)}(y_0^{\ell-1}\mid x)W_{u_0^{i-1}}^{(i)}(y_0^{\ell-1}\mid x')}\\
&=q^{i}\sum_{y_0^{\ell-1}} \sqrt{W^{(i)}(y_0^{\ell-1},u_0^{i-1}\mid x)W^{(i)}(y_0^{\ell-1},u_0^{i-1}\mid x')}\\
&=\frac1{q^{\ell-1-i}}\sum_{y_0^{\ell-1}}
\sqrt{\sum_{v_{i+1}^{\ell-1},w_{i+1}^{\ell-1}}
 W^{\ell}\left(y_0^{\ell-1}\mid g(u_0^{i-1}, x, v_{i+1}^{\ell-1})\right)
W^{\ell}\left(y_0^{\ell-1}\mid g(u_0^{i-1}, x', w_{i+1}^{\ell-1})\right)}\\
&\le \frac1{q^{\ell-1-i}}\sum_{y_0^{\ell-1}} \sum_{v_{i+1}^{\ell-1},w_{i+1}^{\ell-1}}
\sqrt{W^{\ell}\left(y_0^{\ell-1}\mid g(u_0^{i-1}, x, v_{i+1}^{\ell-1})\right)
W^{\ell}\left(y_0^{\ell-1}\mid g(u_0^{i-1}, x', w_{i+1}^{\ell-1})\right)}\\
&\le \frac1{q^{\ell-1-i}}
\sum_{v_{i+1}^{\ell-1},w_{i+1}^{\ell-1}}
Z_\text{max}(W)^{D_{x,x'}^{(i)}(u_0^{i-1})}\\
&= q^{\ell-1-i} Z_\text{max}(W)^{D_{x,x'}^{(i)}(u_0^{i-1})}
\end{align*}
The first inequality is obtained as follows.
\begin{align*}
Z_{x,x'}(W_{ u_0^{i-1}}^{(i)}) &=
\sum_{y_0^{\ell-1}} \sqrt{W_{u_0^{i-1}}^{(i)}(y_0^{\ell-1}\mid x)W_{u_0^{i-1}}^{(i)}(y_0^{\ell-1}\mid x')}\\
&=q^{i}\sum_{y_0^{\ell-1}}
\sqrt{W^{(i)}(y_0^{\ell-1},u_0^{i-1}\mid x)W^{(i)}(y_0^{\ell-1},u_0^{i-1}\mid x')}\\
&=\sum_{y_0^{\ell-1}}
\sqrt{\sum_{v_{i+1}^{\ell-1},w_{i+1}^{\ell-1}}
\frac1{q^{2(\ell-1-i)}}
 W^{\ell}\left(y_0^{\ell-1}\mid g(u_0^{i-1}, x, v_{i+1}^{\ell-1})\right)
W^{\ell}\left(y_0^{\ell-1}\mid g(u_0^{i-1}, x', w_{i+1}^{\ell-1})\right)}\\
&\ge\sum_{y_0^{\ell-1}}
\sum_{v_{i+1}^{\ell-1},w_{i+1}^{\ell-1}} \frac1{q^{2(\ell-1-i)}}
\sqrt{W^{\ell}\left(y_0^{\ell-1}\mid g(u_0^{i-1}, x, v_{i+1}^{\ell-1})\right)
W^{\ell}\left(y_0^{\ell-1}\mid g(u_0^{i-1}, x', w_{i+1}^{\ell-1})\right)}\\
&\ge  \frac1{q^{2(\ell-1-i)}}
Z_\text{min}(W)^{D_{x,x'}^{(i)}(u_0^{i-1})}
\end{align*}
\end{proof}


\begin{corollary}\label{cor:Zminmax}
For $i\in\{0,\dotsc,\ell-1\}$,
\begin{align*}
Z_\text{max}(W^{(i)}) &\le q^{\ell-1-i}Z_\text{max}(W)^{D^{[i]}_\text{min}}\\
\frac1{q^{2\ell-2-i}}Z_\text{min}(W)^{D^{[i]}_\text{max}} &\le Z_\text{min}(W^{(i)}).
\end{align*}
\end{corollary}
From Proposition~\ref{prop:rdirect}, \ref{prop:rconv} and Corollary~\ref{cor:Zminmax}, the following theorems are obtained.
\begin{theorem}
Assume $P(I_\infty(W)\in\{0,1\})=1$.
Let $f(n)$ be an arbitrary function satisfying $f(n)=o(\sqrt{n})$.
It holds
\begin{equation*}
\liminf_{n\to\infty} P\left(Z(W^{(B_1)\dotsc(B_n)}) < 2^{-\ell^{E_1(g) n + t\sqrt{V_1(g)n} + f(n)}}\right) \ge I(W)Q(t)
\end{equation*}
where $E_1(g) = (1/\ell) \sum_i \log_\ell D^{[i]}_\text{\rm min}$ and
where $V_1(g) = (1/\ell) \sum_i (\log_\ell D^{[i]}_\text{\rm min}-E_1(g))^2$.

When $Z_\text{\rm min}(W)>0$,
\begin{equation*}
\limsup_{n\to\infty} P\left(Z(W^{(B_1)\dotsc(B_n)}) < 2^{-\ell^{E_2(g) n + t\sqrt{V_2(g)n} + f(n)}}\right) \le I(W)Q(t)
\end{equation*}
where $E_2(g) = (1/\ell) \sum_i \log_\ell D^{[i]}_\text{\rm max}$ and
where $V_2(g) = (1/\ell) \sum_i (\log_\ell D^{[i]}_\text{\rm max}-E_2(g))^2$.
\end{theorem}

\begin{theorem}
Assume that $g$ is a linear kernel represented by a matrix $G$
and that $P(I_\infty(W)\in\{0,1\})=1$.
Let $f(n)$ be an arbitrary function satisfying $f(n)=o(\sqrt{n})$.
It holds
\begin{equation*}
\liminf_{n\to\infty} P\left(Z(W^{(B_1)\dotsc(B_n)}) < 2^{-\ell^{E(G) n + t\sqrt{V(G)n}+f(n)}}\right) \ge I(W)Q(t).
\end{equation*}

When $Z_\text{\rm min}(W)>0$,
\begin{equation*}
\limsup_{n\to\infty} P\left(Z(W^{(B_1)\dotsc(B_n)}) < 2^{-\ell^{E(G) n + t\sqrt{V(G)n}+f(n)}}\right) \le I(W)Q(t).
\end{equation*}
\end{theorem}

\section{Reed-Solomon kernel}
Assume that $\mathcal{X}$ is a field and that $\alpha\in\mathcal{X}$ is its primitive element.
For a non-zero element $\gamma\in\mathcal{X}$, let
\begin{equation*}
G=
\begin{bmatrix}
1&1&\dotsc&1&1&0\\
\alpha^{(q-2)(q-2)}&\alpha^{(q-3)(q-2)}&\dotsc&\alpha^{q-2}&1&0\\
\alpha^{(q-2)(q-3)}&\alpha^{(q-3)(q-3)}&\dotsc&\alpha^{q-3}&1&0\\
\vdots&\vdots&\dotsc&\vdots&\vdots&\vdots\\
\alpha^{q-2}&\alpha^{q-3}&\dotsc& \alpha&1&0\\
1&1& \dotsc& 1&1&\gamma\\
\end{bmatrix}.
\end{equation*}
When $q$ is prime, channel polarization phenomenon occurs for any $\gamma\ne 0$.
When $\gamma$ is a primitive element of $\mathcal{X}$, channel polarization phenomenon occurs for any field $\mathcal{X}$.
We call $G$ a Reed-Solomon kernel since its submatrix which consists of $i$-th row to $(q-1)$-th row is a generator matrix of a generalized
Reed-Solomon code for any $i\in\{0,\dotsc,q-1\}$~\cite{macwilliams1988tec}.
Since generalized Reed-Solomon codes are maximum distance separable (MDS) codes, it holds $D^{[i]}=i+1$.
Hence, the exponent of Reed-Solomon kernel is $(1/\ell) \log_\ell(\ell!)$ where $\ell=q$.
Since
\begin{equation*}
\frac1\ell \sum_{i=0}^{\ell-1}\log_\ell (i+1)
\ge
\frac1{\ell\log_\mathrm{e}\ell}\int_1^\ell \log_\mathrm{e} x \mathrm{d}x
= 1- \frac{\ell-1}{\ell\log_\mathrm{e}\ell}
\end{equation*}
the exponent of the Reed-Solomon kernel tends to 1 as $\ell=q$ tends to infinity.
The exponent of the Reed-Solomon kernel of size $2^2$ is $\log 24/(4\log 4)\approx 0.57312$.
In~\cite{korada2009pcc}, the authors showed that, by using large kernels, the exponent can be improved,
and found the best matrix of size 16 whose exponent is about 0.51828.
The exponent of the Reed-Solomon kernel on $\mathbb{F}_4$ of size 4 is larger than
the largest exponent of binary matrices of size 16.

The Reed-Solomon kernel can be regarded as a natural generalization of the $2\times 2$ matrix~\eqref{eq:2x2}.
Note that a generator matrix of the $r$-th order $q$-ary Reed-Muller code of length $q^n$ is constructed
by choosing rows 
\begin{equation*}
\left\{j\in\{0,\dotsc,q^n-1\}~\bigm|~\sum_{i=1}^n b_i(j) \ge (q-1)n - r\right\}
\end{equation*}
from $G^{\otimes n}$ where $b_i(j)$ is the $i$-th element of $q$-ary expansion of $j$.
The relation between binary polar codes and binary Reed-Muller codes was mentioned by Ar{\i}kan~\cite{5075875}, \cite{4542778}.


\backmatter
\chapter*{Summary}\label{chap:sum}
In the thesis, we have seen the channel polarization phenomenon and polar codes.
It is shown that polar codes are constructed with linear complexity in the blocklength for symmetric B-DMC\@.
The channel polarization phenomenon on $q$-ary channels has also been considered.
We see sufficient conditions of kernels on which the channel polarization phenomenon occurs.
We also see that the Reed-Solomon kernel is a natural generalization to $q$-ary alphabet of the $2\times 2$ matrix~\eqref{eq:2x2} as a binary matrix.
The exponent of the Reed-Solomon kernel tends to 1 as $q$ tends to infinity.
The exponent of the Reed-Solomon kernel of size $2^2$ is larger than the largest exponent for binary matrices of size 16.

\bibliographystyle{IEEEtranS}
\bibliography{IEEEabrv,ldpc}

\end{document}